\documentclass[a4paper,onecolumn,11pt,unpublished]{quantumarticle}
\pdfoutput=1
\usepackage[english]{babel}
\usepackage[T1]{fontenc}
\usepackage{amsmath,amssymb,amsthm,booktabs}
\usepackage{microtype}
\usepackage[hidelinks]{hyperref}
\newtheoremstyle{articlethm}{\topsep}{\topsep}{\normalfont}{\parindent} {\itshape}{.}{ }{}
\theoremstyle{articlethm}
\newtheorem{theorem}{Theorem}
\newtheorem{proposition}{Proposition}
\newtheorem{lemma}{Lemma}
\newtheorem{corollary}{Corollary}
\theoremstyle{remark}

\DeclareMathOperator{\Var}{Var}
\DeclareMathOperator{\Cov}{Cov}
\DeclareMathOperator{\tr}{tr}
\DeclareMathOperator{\clip}{clip}
\newcommand{\PP}{\mathbb P}
\newcommand{\E}{\mathbb E}
\newcommand{\R}{\mathbb R}
\newcommand{\ket}[1]{\lvert #1\rangle}
\newcommand{\bra}[1]{\langle #1\rvert}
\newcommand{\ip}[2]{\langle #1,#2\rangle}
\newcommand{\norm}[1]{\lVert #1\rVert}
\newcommand{\abs}[1]{\lvert #1\rvert}
\newcommand{\Nzero}{\mathcal N_0}
\newcommand{\one}{\boldsymbol 1}

\begin{document}

\title{Quantum Readout Complexity for Normalized Linear Functionals}
\author{Roi Vence-Nogueira}
\affiliation{School of Mathematical Sciences, Peking University}
\date{October 2026}
\maketitle

\begin{abstract}
Quantum numerical algorithms often encode a positive solution $q$ in the amplitudes of a quantum state, while the required scientific output is a normalized linear functional of $q$. Basis measurements produce probabilities determined by squared solution values, while the functional uses the solution values themselves. We study the complexity of this output-interface problem for $F_\phi(q)=\int\phi q\,\mathrm d\nu/\int q\,\mathrm d\nu$ and identify a target-dependent susceptibility $\chi_\phi$ that governs the amount of copies of the solution state needed. For a fixed positive reference state and a neighborhood around it, the optimal number of unknown-state copies is \(\Theta\!\left(\chi_0\varepsilon^{-2}\log(1/\beta)\right)\), where \(\varepsilon\) is the absolute error tolerance and \(\beta\) is the allowed failure probability. We show that two fixed single-copy measurements attain the same rate. We then analyze how this baseline changes when the interface retains classical information, uses adaptive allocation, supplies coherent access, or restricts the allowed solution family. Therefore $\chi_\phi$ quantifies a central cost of using a solution encoded in an $L^2$-normalized quantum state to estimate an $L^1$-normalized linear observable.
\end{abstract}

\section{Introduction}
\label{sec:intro}

Quantum algorithms, such as those for linear systems and differential equations, often prepare a state whose amplitudes encode the solution~\cite{hhl,montanaro}. Recovering the full solution can be costly, so the goal is often a particular quantity instead. This quantity might be a spatial average, an expected payoff, or the fraction of a physical quantity in a certain region. Obtaining such an answer requires a suitable readout procedure. We study the problem of readout from state copies alone, provide a local copy complexity, and then analyze how the complexity changes when additional classical information or coherent access to the preparation is provided.

For a solution with positive values $q_i$ on a grid with quadrature weights $\nu_i>0$, the amplitude encoding is
\begin{equation}\label{eq:finite-encoding}
 \ket{x(q)}=\frac{\sum_i\sqrt{\nu_i}q_i\ket{i}}
                      {\sqrt{\sum_i\nu_iq_i^2}}.
\end{equation}
The quantities considered are normalized linear functionals,
\begin{equation}\label{eq:opening-general-target}
 F_\phi(q)=\frac{\sum_i\nu_i\phi_iq_i}{\sum_i\nu_iq_i},
 \qquad \abs{\phi_i}\leq1.
\end{equation}
where $\phi$ specifies the weighting. The state is normalized in $L^2$ norm, while the target is normalized by the $L^1$ mass of $q$. This mismatch and its consequences are the object of our study.

To illustrate this problem, consider estimating a forward probability density. Let $q_i\approx p(x_i,T)$, where $p(\cdot,T)$ is the density of a certain random variable $X_T$. Setting $\phi_i=\one_A(x_i)$ gives the probability of a region $A$:
\begin{equation}\label{eq:opening-target}
 \PP\left[X_T\in A\right]\approx
 F_{\one_A}(q)=\frac{\sum_i\nu_iq_i\one_A(x_i)}{\sum_i\nu_iq_i}.
\end{equation}
To instead compute an expectation, we could replace the indicator by a payoff.

In all of these cases, basis measurement of~\eqref{eq:finite-encoding} sample from $Q_i\propto\nu_iq_i^2$, but the target average has $P_i\propto\nu_iq_i$. If the solution was encoded with amplitudes $\sqrt{P_i}$, this would allow direct sampling from $P$. However, when the solution itself is encoded in the amplitudes, recovering the $P$-weighted average requires a different readout procedure. We ask how many copies are generally necessary, and whether a fixed measurement procedure can attain that cost? We answer both questions on a local neighborhood for every bounded real $\phi$, and then analyze these results under several enrichments of the information of the output interface. These results apply, in particular, to weighted grids and continuous fields on finite-measure spaces, no matter how the solution was produced.

\subsection{What information reaches the readout stage?}
In the copy model, the readout procedure consists of identical copies of the normalized states and known grid weights and observable $\phi$. Joint quantum operations, known reference states, and classical computation are all allowed, and we only count unknown-state copies. The measurements depend on known data, and neither the classical input data nor a coherent preparation unitary are supplied, although these restrictions are relaxed in subsequent results.

A similar problem arises in composable quantum subroutines. Modular quantum computation requires a producer--client interface with agreement about which concrete data and operations are exchanged~\cite{modular-qc}. The client doesn't have access to how the state was produced, and the readout interface determines which downstream tasks are possible.

\subsection{A condition number for the requested output}
Different amplitude states can be close to each other while producing significantly different physical averages for a certain $\phi$. We show that readout is difficult when the averages changes rapidly in directions that change the state only slightly. For $\mu=\E_{P}[\phi]$, we quantify this sensitivity by
\begin{equation}\label{eq:intro-closed}
 \chi_\phi(q)=
 \frac{\sum_i\nu_iq_i^2}{(\sum_i\nu_iq_i)^2}
       \sum_i\nu_i(\phi_i-\mu)^2.
\end{equation}
Section~\ref{sec:model} establishes this interpretation of the formula.

For $W=\sum_i\nu_i$, let $U_i=\nu_i/W$ and $\kappa=(\sum_i\nu_iq_i)/\sqrt{W\sum_i\nu_iq_i^2}$. Then we establish the decomposition:
\begin{equation}\label{eq:intro-factorization}
 \chi=\underbrace{\kappa^{-2}}_{\text{representation concentration}}
       \underbrace{\E_{U}\left[(\phi-\mu)^2\right]}_{\text{target sensitivity}}.
\end{equation}
Which shows that both concentration and sensitivity affect the readout difficulty.

\subsection{What the main theorem establishes}
For a positive reference solution $q_0$ consider the set of inputs whose normalized states lie within distance $\kappa_0/4$ from $x_0=q_0/\norm{q_0}_2$. The state $x_0$ is the center of the allowed neighborhood $\Nzero$. For a fixed center $q_0$, let $K^*(\varepsilon,\beta;\Nzero)$ be the smallest deterministic copy budget that a procedure needs to estimate the target with error at most $\varepsilon$, with failure probability at most $\beta$, for every input in $\Nzero$. For $0<\varepsilon\leq\kappa_0\sqrt{\chi_0}/12$, Theorem~\ref{thm:main} shows that:
\[
 K^*(\varepsilon,\beta;\Nzero)
 =\Theta\!\left(\chi_0\varepsilon^{-2}\log(1/\beta)\right),
 \qquad 0<\beta\leq1/3.
\]
Our lower bound covers all estimators, including biased ones, and all measurements on the available copies. The upper bound uses two fixed settings determined by $\nu$ and $\phi$, where one is used to estimate a numerator and the other its normalization. The settings don't require the reference solution nor the reconstruction of the unknown $q$, and the reference is only needed to predetermine the budget. We also relax this requirement and show the concrete cost without this predetermination. Therefore our results characterize the state-only interface, showing that $\chi_0$ is the corresponding finite-sample copy condition number.

Corollary~\ref{cor:class-lower} extends the lower bound to any output class containing such a neighborhood, using the largest susceptibility among all the admissible centers. Proposition~\ref{prop:basis-gap} shows that at fixed accuracy, basis sampling can require $\Omega(\sqrt D)$ copies, while our two measurements require only a constant number.

Theorem~\ref{thm:main} and Proposition~\ref{prop:basis-gap} are the main state-only results. Sections~\ref{sec:adaptive}--\ref{sec:queries} analyze the complexity under adaptive allocation, grid refinement, retained classical information, and retained coherent preparation. We then analyze the concrete consequences for numerical algorithms. Appendix~\ref{app:main-proof} provides the proofs for all the results, and Appendix~\ref{app:supporting} extends the main results to signed and complex states, and analyzes several numerical applications.

\subsection{Relation to prior work}
Quantum PDE and linear-system algorithms have long distinguished preparing a normalized solution state from estimating a classical quantity of interest~\cite{hhl,montanaro,linden,deiml}. Coherent readout and direct-observable PDE algorithms~\cite{li-elliptic,li-parabolic} use resources beyond the normalized-state copies. Our work studies the optimal cost when those copies are the only information delivered to the readout procedure.

Linear-combination-of-unitaries (LCU) methods also address observable extraction. Chakraborty's single-ancilla protocol~\cite{chakraborty-lcu} uses the terms of an operator known decomposition to estimate normalized quadratic observables by repeated measurements and classical averaging. Our minimax bounds concern readout from copies of the resulting state.

The importance-weighted expression for $\chi$ coincides with the asymptotic variance of self-normalized importance sampling~\cite{geweke,owen}, although in our work the weights depend on the unknown solution and are not available to the estimator. Quantum semiparametric and pure-state theory provide derivative bounds and pointwise efficiency results~\cite{tsang,matsumoto}. Lahiry--Nussbaum~\cite{lahiry} establish asymptotic minimax squared risk for functionals linear in the density operator on shrinking neighborhoods. For finite samples, Meyer et al.~\cite{meyer} framework studies arbitrary state properties, with bounds on the number of copies through R\'enyi divergence between states.

Chen--Gong--Zhou~\cite{chen} give a finite-sample Fisher-information characterization for estimation of known observable families $\tr(O_i\rho)$. Their local estimator does not cover the present pure-state ratio experiment.

Strong measurements can recover weak values~\cite{cohen-pollak}. Other overlap-based readout protocols include controlled-SWAP, circuits based on Bargmann invariants, and fidelity estimation~\cite{chiribella,wagner,wang-zhang}. Our coherent-access procedure uses standard amplitude estimation~\cite{bhmt,alase}.

Our contribution is their explicit finite-sample resolution for this normalized-linear positive-amplitude interface: the positive-part construction produces separated targets with divergence $O(\varepsilon^2/\chi_0)$ inside a fixed neighborhood without a minimum-density assumption, while a centered variance identity lets two fixed settings attain the same scale uniformly. Together these give an optimal minimax copy complexity for biased estimators and arbitrary collective measurements.

\section{The experiment and its condition number}\label{sec:model}

We work on $(\Omega,\nu)$ with $0<W=\nu(\Omega)<\infty$ and separable $L^2(\nu)$, allowing us to treat weighted grids and continuous fields within the same framework. We assume an almost everywhere positive $q\in L^2(\nu)$, and a known real-valued measurable $\phi$ satisfying $\abs{\phi}\leq1$. Define
\begin{equation}\label{eq:model}
 \begin{aligned}
 Z&=\int q\,\mathrm d\nu,& S&=\int q^2\,\mathrm d\nu,\\
 \mathrm dP&=\frac qZ\,\mathrm d\nu,& \mathrm dQ&=\frac{q^2}S\,\mathrm d\nu,\\
 x_q&=\frac q{\sqrt S},&
 \mu&=F_\phi(q)=\int\phi\,\mathrm dP.
 \end{aligned}
\end{equation}
Given that $\nu(\Omega)<\infty$, we have $0<Z<\infty$. The quantum system is the complexification of $L^2(\nu)$, with pure density operator $\ket{x_q}\bra{x_q}$. We choose the state vector with positive real amplitudes. For a finite set with atomic weights $\nu_i$, the orthonormal-coordinate vector is exactly~\eqref{eq:finite-encoding}. Both state and target are invariant under positive rescaling of $q$.

We assume reception of independent copies of this state. All POVMs, joint or adaptive processing, and classical computation are permitted. We count only unknown-state copies, while the processing of known measurements and reference states not included in the complexity count. A deterministic copy bound gives a maximum on every possible execution. Our main result doesn't use information from a classical description of $q$, samples from $P$, or coherent access to its preparation.

Let $U=\nu/W$, let $v=\one/\sqrt W$ be the known unit reference vector, and let $\Phi$ be multiplication by $\phi$. Then
\begin{equation}\label{eq:ratio}
 \kappa=\ip{v}{x_q}=\frac Z{\sqrt{WS}}>0,
 \qquad F_\phi(q)=\frac{\ip{v}{\Phi x_q}}{\ip{v}{x_q}}.
\end{equation}
For a normalized state $x$ in the positive neighborhood, write
\begin{equation}
 \mathcal F(x):=\frac{\ip{v}{\Phi x}}{\ip{v}{x}},
 \qquad \mathcal F(x_q)=F_\phi(q).
\end{equation}
This identity gives bounded finite-rank readout observables even when the ambient Hilbert space is infinite-dimensional.

To quantify the local difficulty of measurement, perturb $q$ multiplicatively as $q_t=q(1+th)$, for real $h$. The physical average changes at speed $\Cov_P(\phi,h)$, while the speed of the normalized state is $\sqrt{\Var_Q(h)}$. We measure sensitivity by the ratio of squared target speed to squared state speed and take its supremum:
\begin{equation}\label{eq:intro-chi}
 \begin{aligned}
 \chi_\phi(q)&=\sup_h\frac{\Cov_P(\phi,h)^2}{\Var_Q(h)}\\
 &=\E_Q\!\left[\left(\frac{\mathrm dP}{\mathrm dQ}\right)^2
       (\phi-\mu)^2\right].
 \end{aligned}
\end{equation}
This is the variance functional of self-normalized importance sampling~\cite{geweke,owen}.

\begin{proposition}[Explicit susceptibility]\label{prop:chi}
Put $w=(S/(Zq))(\phi-\mu)$. The supremum in~\eqref{eq:intro-chi}, taken over bounded real measurable $h$ with $\Var_Q(h)>0$, is
\begin{equation}\label{eq:chi}
 \begin{aligned}
 \chi=\E_Q\!\left[w^2\right]
 &=\frac S{Z^2}\int(\phi-\mu)^2\,\mathrm d\nu\\
 &=\kappa^{-2}\left(\Var_U(\phi)+
       \left(\E_U[\phi]-\mu\right)^2\right).
 \end{aligned}
\end{equation}
The empty supremum is zero. Moreover $\chi=0$ if and only if $\phi$ is constant almost everywhere, in which case zero copies suffice. For every bounded $h$, along $q_t=q(1+th)$,
\begin{equation}\label{eq:derivatives}
 \begin{aligned}
 \dot F(0)&=\Cov_P(\phi,h),&
 \dot x(0)&=x_q\left(h-\E_Q[h]\right),\\
 \norm{\dot x(0)}_2^2&=\Var_Q(h).
 \end{aligned}
\end{equation}
\end{proposition}
The proof is given in Appendix~\ref{app:main-proof}.

When $\chi>0$, in finite dimension the maximizing directions are $h=cw+b\one$ with $c\ne0$. In the continuum a bounded optimum doesn't generally exist, we might approach the same supremum in the limit. Also
\begin{equation}\label{eq:cv}
 \kappa^{-2}=1+\frac{\Var_U(q)}{\left(\E_U[q]\right)^2}.
\end{equation}
For $\phi=\one_A$, with $a_A=U(A)$ and $m=P(A)$, $\chi=\kappa^{-2}\{a_A(1-a_A)+(a_A-m)^2\}$.

\subsection{Target sensitivity and concentration}\label{sec:target-comparison}
Regional averages show why the requested observable matters as much as the solution. Consider a unit-volume domain with a uniform background and a narrow component of fixed mass $\alpha\in(0,1)$:
\begin{equation}\label{eq:bump}
 q=(1-\alpha)+\frac\alpha\delta\one_B,\qquad \nu(B)=\delta.
\end{equation}
Then $Z=1$ and $S=1+\alpha^2(\delta^{-1}-1)=\Theta(\delta^{-1})$. For a fixed region $A\supset B$ of volume $a_A\in(0,1)$ and target $\phi=\one_A$,
\begin{align*}
 F_\phi(q)&=(1-\alpha)a_A+\alpha,\\
 \chi_\phi&=S\bigl[a_A(1-a_A)+\alpha^2(1-a_A)^2\bigr]
             =\Theta(\delta^{-1}).
\end{align*}
The admissible local accuracy remains of constant order because
$\kappa\sqrt{\chi_\phi}=\sqrt{\chi_\phi/S}=\Theta(1)$. Therefore a fixed-mass component whose width shrinks can make the copy cost grow inversely with its width.

Now choose a second region $C$ disjoint from $B$, with $\nu(C)=\delta$, and let $\psi=\one_C$. For the same state,
\begin{align*}
 F_\psi(q)&=(1-\alpha)\delta,\\
 \chi_\psi&=S\bigl[\delta-(1-\alpha^2)\delta^2\bigr]=\Theta(1).
\end{align*}
Both targets are nonconstant indicators with norm one, but their local copy rates differ by $\Theta(\delta^{-1})$. For sufficiently small $\delta$, this holds in $0<\varepsilon\leq\sqrt{\delta-(1-\alpha^2)\delta^2}/12$, and Theorem~\ref{thm:main} gives rates $\Theta(\delta^{-1}\varepsilon^{-2}\log(1/\beta))$ for $\phi$ and $\Theta(\varepsilon^{-2}\log(1/\beta))$ for $\psi$. Concentration alone therefore does not determine the cost of the requested output.

\section{Sharp local copy complexity}\label{sec:copies}

The susceptibility describes infinitesimal changes. Our main result shows that it also determines a finite copy budget, uniformly over a whole neighborhood.

Choose a positive reference solution $q_0$ and use subscript $0$ for all its quantities, writing $\kappa_0:=\kappa(q_0)$ and $\chi_0:=\chi_\phi(q_0)$ explicitly. We call $q_0$ the center, write this neighborhood as $\Nzero(q_0)$, and allow every positive input whose normalized state lies in it:
\begin{equation}\label{eq:neighborhood}
 \Nzero=\{q>0\text{ a.e.}:q\in L^2(\nu),\ 
                   \norm{x_q-x_0}_2\leq\kappa_0/4\},
\end{equation}
where the radius is independent of the requested accuracy.

\subsection*{What is the optimal copy budget?}
Fix an error tolerance $\varepsilon>0$ and a failure probability $\beta$. A protocol with budget $k$ may use at most $k$ copies on every run. It must return an estimate $T$ satisfying
\begin{equation}\label{eq:uniform-success}
 \sup_{q\in\Nzero}
 \PP_q\left[\abs{T-F_\phi(q)}>\varepsilon\right]\leq\beta.
\end{equation}
That is, worst-case in $\Nzero$, the region where estimator error is bigger than $\varepsilon$ must have probability less than $\beta$. We define
\begin{equation}\label{eq:copy-budget}
 K^*(\varepsilon,\beta;\Nzero)
 =\min\{k\in\mathbb N_0:\text{such a protocol exists}\}.
\end{equation}
Therefore $K^*$ is the smallest sufficient worst-case copy budget, optimized over all measurements and estimators. For another known input class $\mathcal C$, the notation $K^*(\varepsilon,\beta;\mathcal C)$ has the same meaning with $\Nzero$ replaced by $\mathcal C$.

The known center can be used to choose this budget. The measurement settings and ratio estimator require only $\nu$ and $\phi$.

\begin{theorem}[Local minimax copy complexity]\label{thm:main}
Under the assumptions of Section~\ref{sec:model}, suppose $\chi_0>0$. If
\begin{equation}\label{eq:accuracy}
 0<\varepsilon\leq\frac{\kappa_0\sqrt{\chi_0}}{12},
 \qquad 0<\beta\leq1/3,
\end{equation}
then
\begin{equation}\label{eq:main-bounds}
 \begin{aligned}
 \frac{\chi_0}{9\varepsilon^2}
       \log\frac1{4\beta(1-\beta)}
 &\leq K^*(\varepsilon,\beta;\Nzero),\\
 K^*(\varepsilon,\beta;\Nzero)
 &\leq 2m(\beta)\left\lceil\frac{48\chi_0}{\varepsilon^2}\right\rceil,
 \end{aligned}
\end{equation}
where
\[
 m(\beta)=
 \begin{cases}
 1, & \beta=1/3,\\
 2\left\lceil 9\log(1/\beta)-\tfrac12\right\rceil+1,
    & 0<\beta<1/3.
 \end{cases}
\]
Therefore \[K^*(\varepsilon,\beta;\Nzero)=\Theta(\chi_0\varepsilon^{-2}\log(1/\beta))\] with universal constants.
\end{theorem}

\subsection{Two fixed measurements that attain the bound}
Two predetermined single-copy measurements attain the upper bound. They depend only on the known measure and observable. Consider:
\begin{equation}\label{eq:main-settings}
 \begin{aligned}
 a&=\E_U[\phi],& z&=(\Phi-aI)v,\\
 \sigma^2&=\norm{z}^2=\Var_U(\phi),\\
 C&=\frac{\ket v\bra z+\ket z\bra v}{2},&
 D&=\ket v\bra v .
 \end{aligned}
\end{equation}
The projector $D$ tests for the uniform reference $v$, while $C$ couples that reference to the centered target vector $z$. Their measurement outcomes satisfy
\begin{equation}\label{eq:main-moments}
 \begin{aligned}
 \E[C]&=\kappa^2(\mu-a),& \E[D]&=\kappa^2,\\
 \E[C^2]&=\norm{Cx_q}_2^2=\frac{\kappa^4\chi}{4}.
 \end{aligned}
\end{equation}
Therefore $\mu=a+\E[C]/\E[D]$, which gives a strong-measurement representation of a weak value~\cite[Sec.~II]{cohen-pollak}. Replacing these means by independent sample averages gives the estimator. After division by the squared denominator, the numerator's second moment contributes $\chi/4$.

Within the neighborhood, $\kappa\geq3\kappa_0/4$ and
$\chi\leq(16/9)^2\chi_0$, so the same copy budget works for every allowed state. Repeating the estimator independently and taking the median shows logarithmic dependence on the failure probability. In Appendix~\ref{app:main-proof} we give the estimator, its variance bound, and the constants.

\subsection{Why fewer copies are not sufficient}
To prove the lower bound, we construct two states in $\Nzero$ that are hard to distinguish but whose target values differ by more than $2\varepsilon$. An estimator that is accurate to $\varepsilon$ on both states can be used to distinguish them. Write $r=v/\kappa_0$ and $d=(\Phi-\mu_0I)r$, so $\|d\|_2^2=\chi_0$ and $d\perp x_0$. Splitting $d=d_+-d_-$ and normalizing $x_0+td_+$ and $x_0+td_-$ separately preserves strict positivity. With $t=3\varepsilon/\chi_0$, the two alternatives remain in $\Nzero$, their target values differ by more than $2\varepsilon$, and their fidelity has:
\[
 -\log |\langle x_+,x_-\rangle|^2
 \leq \frac{9\varepsilon^2}{\chi_0}.
\]
Any uniformly successful estimator must distinguish these two target values. Binary quantum testing~\cite[Corollary 4]{meyer} and the Helstrom bound~\cite{helstrom} then force the stated copy cost, even for biased estimators and collective measurements. We show how to construct this pair inside the positive neighborhood without a lower density floor. Appendix~\ref{app:main-proof} proves its separation and fidelity bounds.

\subsection{Copy bounds for other output families}
\label{sec:modulus}

Theorem~\ref{thm:main} gives finite-sample copy bounds on $\Nzero$. For a different output family $\mathcal C$, the bounds depend on which inputs remain admissible. Containing $\Nzero$ preserves the lower bound, while being contained in $\Nzero$ preserves the upper bound. A lower bound for a restricted family can be obtained by constructing hard pairs within that family.

Let $\mathcal C$ be a nonempty class of positive inputs, with the same copy access and worst-case guarantee as in Section~\ref{sec:copies}. Define the separated pairs
\[
 \mathcal P_\varepsilon(\mathcal C)
 =\{(q,q')\in\mathcal C^2:
       \abs{F_\phi(q)-F_\phi(q')}>2\varepsilon\}.
\]
A gap greater than $2\varepsilon$ is required, as targets exactly $2\varepsilon$ apart can both be estimated by their midpoint with error at most $\varepsilon$. The smallest fidelity divergence among these pairs is
\begin{equation}\label{eq:moduli}
 \Lambda_{\mathcal C}(\varepsilon)
 =\inf_{(q,q')\in\mathcal P_\varepsilon(\mathcal C)}
   \left[-\log\abs{\ip{x_q}{x_{q'}}}^2\right].
\end{equation}
These pairs give lower bounds through quantum hypothesis testing~\cite{donoho,meyer}.

\begin{proposition}[Readout bounds from separated targets]
\label{prop:modulus}
If $\mathcal P_\varepsilon(\mathcal C)$ is nonempty and $0<\beta<1/2$, then
\begin{equation}\label{eq:global-lower}
 K^*(\varepsilon,\beta;\mathcal C)
 \geq
 \frac{\log(1/[4\beta(1-\beta)])}
      {\Lambda_{\mathcal C}(\varepsilon)},
\end{equation}
with extended-value conventions.
\end{proposition}
\begin{proof}
An estimator that succeeds on both inputs of a separated pair gives a test between their states. Apply~\eqref{eq:helstrom}
to each pair and take the supremum of the resulting lower bounds.
\end{proof}

If $\Lambda_{\mathcal C}(\varepsilon)=0$, no finite uniform copy cap exists for $\beta<1/2$. If $\mathcal P_\varepsilon(\mathcal C)$ is empty, the target values span an interval of length at most $2\varepsilon$ and its midpoint requires zero copies.

Theorem~\ref{thm:main} supplies hard pairs whenever $\mathcal C$ contains an admissible neighborhood. Taking the largest susceptibility among its admissible centers gives the following class-level bound.

\begin{corollary}[Class-level state-only lower bound]
\label{cor:class-lower}
For fixed $\varepsilon>0$, define
\begin{equation}\label{eq:admissible-centers}
 \mathcal A_{\varepsilon}(\mathcal C)
 =\left\{q_0\in\mathcal C:
 \mathcal N_0(q_0)\subseteq\mathcal C,\quad
 \varepsilon\leq\frac{\kappa_0\sqrt{\chi_0}}{12}\right\},
\end{equation}
where $\mathcal N_0(q_0)$ is the neighborhood~\eqref{eq:neighborhood} centered at $q_0$. Then, for $0<\beta\leq1/3$,
\begin{equation}\label{eq:class-lower}
 K^*(\varepsilon,\beta;\mathcal C)
 \geq
 \frac{\log\!\left(1/[4\beta(1-\beta)]\right)}
      {9\varepsilon^2}
 \sup_{q_0\in\mathcal A_{\varepsilon}(\mathcal C)}
       \chi_\phi(q_0),
\end{equation}
with the convention that the supremum is zero if the set is empty.
\end{corollary}
\begin{proof}
For every $q_0\in\mathcal A_{\varepsilon}(\mathcal C)$, a successful protocol on $\mathcal C$ also succeeds uniformly on $\mathcal N_0(q_0)$. Theorem~\ref{thm:main} therefore gives the corresponding lower bound. Taking the supremum proves~\eqref{eq:class-lower}.
\end{proof}

For a family $\mathcal C\subseteq\Nzero$, the two fixed measurements instead give the upper bound
\[
 K^*(\varepsilon,\beta;\mathcal C)
 =O\!\left(\chi_0\varepsilon^{-2}\log(1/\beta)\right)
\]
in the accuracy and confidence range of Theorem~\ref{thm:main}. If, for each such $\varepsilon$, the family contains $q_+,q_-$ with
\begin{equation}\label{eq:restricted-pair}
 \begin{aligned}
 \abs{F_\phi(q_+)-F_\phi(q_-)}&>2\varepsilon,\\
 -\log\abs{\ip{x_{q_+}}{x_{q_-}}}^2
 &\leq c_{\mathrm{div}}\frac{\varepsilon^2}{\chi_0},
 \end{aligned}
\end{equation}
for a constant $c_{\mathrm{div}}$ independent of $\varepsilon$, Proposition~\ref{prop:modulus} gives the matching lower bound. Thus
\begin{equation}\label{eq:restricted-rate}
 K^*(\varepsilon,\beta;\mathcal C)
 =\Theta\!\left(\frac{\chi_0}{\varepsilon^2}
                         \log\frac1\beta\right),
\end{equation}
up to constants depending only on $c_{\mathrm{div}}$. Containing a whole neighborhood supplies hard pairs by Corollary~\ref{cor:class-lower}.

A large susceptibility alone does not ensure a large copy cost on a restricted family. For example, consider $\nu_i=1$, $q=(1,\delta,\ldots,\delta)$ in dimension $n+1$, and $\phi=(0,1,\ldots,1)$, restrict $0<\delta\leq\varepsilon/(4n)$, with $0<\varepsilon\leq1$. Then
\begin{equation}\label{eq:thin}
 F_\phi(q)=\frac{n\delta}{1+n\delta}\leq\varepsilon/4,
 \qquad
 \chi_\phi(q)
 =n\left(\frac{1+n\delta^2}{(1+n\delta)^2}\right)^2
 \geq n/16.
\end{equation}
Returning zero succeeds without measurements because this family excludes the hard pairs. Applying the lower bound to a numerical solver therefore requires hard pairs among its possible outputs. Section~\ref{sec:applications} examines this requirement in concrete numerical settings.

\subsection{The measurement basis}
Can the same copy bound be achieved by measuring in the computational basis, and processing the samples classically? The following proposition shows that it cannot. On the same neighborhood, at fixed accuracy and confidence, every basis-only estimator requires $\Omega(\sqrt D)$ copies, while our two settings need a copy budget independent of $D$.

\begin{proposition}[Basis-only finite-sample separation]
\label{prop:basis-gap}
Let $D$ be divisible by four, use counting measure on $D$ coordinates, and take $q_0=\one$ and $\phi=\one_A$, where $A$ is the first half of the coordinates. Then $\kappa_0=1$ and $\chi_0=1/4$. On the neighborhood~\eqref{eq:neighborhood}, every basis-only estimator that achieves accuracy $\varepsilon_*=1/128$ and failure probability at most $1/3$ with deterministic sample cap $k$ satisfies
\begin{equation}\label{eq:basis-gap}
 k(k-1)\geq D/3.
\end{equation}
The bound already holds on a subclass with $1/(2D)\leq Q_i\leq3/(2D)$. The two settings in~\eqref{eq:main-settings} attain the same accuracy and confidence with a copy cap independent of $D$.
\end{proposition}

The lower bound uses two families with different target values whose basis samples are hard to distinguish until a hidden pair of coordinates is sampled more than once~\cite{paninski}. Appendix~\ref{app:basis-separation} gives the proof and its continuum consequence.

\subsection{Zeros, signs, and phases}
We now consider relaxing the positivity assumption on $q$.
\paragraph{Zeros and known support}
For non-negative $q$, if a fixed support is known, restrict $\nu$ to that support and use the theorem unchanged. In general, the finite-sample bounds of Theorem~\ref{thm:main}, with the same constants, hold for a non-negative center and its non-negative neighborhood if $\chi$ is replaced by $\bar\chi=(S/Z^2)\int(\phi-\mu)^2\,\mathrm d\nu$ over the full reference space. The upper-bound identities remain valid, and $x_0+td_\pm$ preserves the non-negativity, so the separation and fidelity estimates above also apply.

\paragraph{Signed and complex solution states.}
Positivity gives the target an interpretation as a physical average, however, the ratio readout also works for signed and complex amplitudes. Let $a\ne0,b$ be vectors in a separable Hilbert space, and define
\begin{equation}\label{eq:general-ratio-main}
 \begin{aligned}
 F_{a,b}(q)&=\frac{\ip{b}{q}}{\ip{a}{q}},\\
 \chi_{a,b}(q)&=
 \frac{\norm{q}^2\norm{b-\overline{F_{a,b}(q)}a}^2}
      {\abs{\ip{a}{q}}^2}.
 \end{aligned}
\end{equation}

For $\ip{a}{q}\ne0$, the ratio is invariant under global phase. Theorem~\ref{thm:general-ratio} gives the same minimax rate using two settings for real data and three for complex data. For signed $q$, the denominator is not the $L^1$ norm. The positive case additionally requires hard pairs that remain positive.

\section{Choosing the budget without a reference solution}\label{sec:adaptive}

In Theorem~\ref{thm:main} we require a reference solution in order to set the deterministic budget. In this section, we show how to obtain a budget without a reference solution through an adaptative procedure, whose copy count is random and bounded in expectation and with high probability, with rare outcomes taking longer to observe.

There are just two probabilities needed to learn: $\eta$, the frequency of $D=1$, and $p$, the frequency of nonzero $C$ outcomes. For $\sigma>0$, write $e=z/\sigma$. Then
\begin{equation}\label{eq:pilot-probabilities}
 \begin{aligned}
 \Pi&=\ket v\bra v+\ket e\bra e,&
 \eta&=\ip{x}{Dx}=\kappa^2,\\
 p&=\ip{x}{\Pi x}
 =\kappa^2\left(1+\frac{(\mu-a)^2}{\sigma^2}\right).
 \end{aligned}
\end{equation}
The identity $C^2=(\sigma^2/4)\Pi$ explains why the nonzero outcomes measure $p$. Since $\chi=\sigma^2p/\eta^2$ and $0<\eta\leq p\leq1$, these two frequencies determine a sufficient budget.

\begin{theorem}[Adaptive allocation]\label{thm:adaptive}
For every $\varepsilon>0$ and $0<\beta\leq1/3$, a protocol using only $C,D$ and known data terminates almost surely, with
\[
 \PP_q\left(\abs{\widehat F-F_\phi(q)}>\varepsilon\right)\leq\beta
\]
for every positive input. Its random copy count $K$ satisfies
\begin{equation}\label{eq:adaptive-order}
 \E_q\!\left[K\right]=O\!\left(\left(\kappa^{-2}+\frac\chi{\varepsilon^2}\right)
                          \log\frac1\beta\right).
\end{equation}
With probability at least $1-\beta/2$, $K$ is bounded by the same order. 
\end{theorem}

\paragraph{How the budget is learned.}

Two independent pilots wait for a preset number of outcomes with $D=1$ and $C\ne0$, respectively. Their waiting times give constant-factor estimates of $\eta$ and $p$~\cite{mendo-hernando}, which allow to determine a conservative budget. Fresh copies then run the ratio estimator with median amplification. In Appendix~\ref{app:pilot} we give the protocol and its bounds.

Learning the budget adds an extra $\kappa^{-2}$ term, which is the cost of waiting for denominator successes. When $\varepsilon\leq\kappa\sqrt{\chi}$ holds, this term is at most $\chi/\varepsilon^2$, so adaptation still preserves the overall readout rate.

\section{What happens under grid refinement?}\label{sec:continuum}

Refining the grid can represent the same underlying field with more precision, using more amplitudes. In this section we show that, under consistent discretization, the target and susceptibility converge to their continuum values.

\begin{proposition}[Consistent discretization]\label{prop:mesh}
Let finite partitions have cells $A_{h,i}$ of positive measure $\nu_{h,i}=\nu(A_{h,i})$. Let $q_h,\phi_h$ be piecewise constant on these cells, with $q_h>0$, $\abs{\phi_h}\leq1$, and $q_h\to q$, $\phi_h\to\phi$ in $L^2(\nu)$. Then
\begin{equation}\label{eq:mesh}
 Z_h\to Z,\quad S_h\to S,\quad \mu_h\to\mu,\quad
 \kappa_h\to\kappa,\quad\chi_h\to\chi.
\end{equation}
If $\chi>0$, every fixed $\varepsilon<\kappa\sqrt\chi/12$ is in the accuracy range of Theorem~\ref{thm:main} for all sufficiently fine partitions.
\end{proposition}
The proof is given in Appendix~\ref{app:mesh-proof}.

A finer grid doesn't necessarily lead to more readout copies. For computing cell averages with a refining partition, Proposition~\ref{prop:mesh} gives $\chi_h\to\chi$ and $F_{\phi_h}(q_h)\to F_\phi(q)$. Therefore the readout cost remains under control, even as the discretization bias vanishes. The total error satisfies
\begin{equation}\label{eq:grid-error-budget}
 \abs{T-F_\phi(q)}
 \leq \abs{T-F_{\phi_h}(q_h)}
       +\abs{F_{\phi_h}(q_h)-F_\phi(q)}.
\end{equation}

The finite-measure assumption is necessary for this result. On unbounded-measure Lebesgue domains, a small $L^2$ perturbation can still carry a significant amount of $L^1$ mass. Even finite $\chi$ at a center doesn't ensure a similar finite uniform copy bound on any state-norm neighborhood. In Appendix~\ref{app:infinite} we give an explicit counterexample.

\section{Readout with retained classical information}\label{sec:extra-info}

Classical information supplied with the state can restore the solution's scale or resolve part of the target. Knowing one norm leaves the normalized-target copy cost unchanged, while known moments can reduce susceptibility. A linear-system example shows how additional records progressively reduce the cost to zero.

\subsection{Known scale and linear observables}
A solver might provide information about the total mass or the $L^2$ norm of its solution. We show that fixing either of these norms separately does not restrict the normalized states in $\Nzero$, so the copy bound is unchanged.

\begin{proposition}[Known scale and linear observables]
\label{prop:fixed-scale}
Under the assumptions of Section~\ref{sec:model}, set $Z_0=\norm{q_0}_1$ and $s_0=\norm{q_0}_2$, and define
\begin{equation}\label{eq:fixed-scale-classes}
 \begin{aligned}
 \Nzero^{(1)}&=\{q\in\Nzero:\norm{q}_1=Z_0\},\\
 \Nzero^{(2)}&=\{q\in\Nzero:\norm{q}_2=s_0\}.
 \end{aligned}
\end{equation}
In each class the fixed norm is public, with all other access as in the copy model. For every $\varepsilon>0$ and $0<\beta<1$,
\begin{equation}\label{eq:fixed-scale-F}
 \begin{aligned}
 K^*(\varepsilon,\beta;\Nzero^{(1)})
 &=K^*(\varepsilon,\beta;\Nzero^{(2)})\\
 &=K^*(\varepsilon,\beta;\Nzero).
 \end{aligned}
\end{equation}
For $L_\phi(q)=\int\phi q\,\mathrm d\nu$, write $K_L^*$ for the analogous copy budget with $L_\phi$ in place of $F_\phi$. Then
\begin{equation}\label{eq:fixed-scale-L}
 K_L^*(\varepsilon,\beta;\Nzero^{(1)})
 =K^*(\varepsilon/Z_0,\beta;\Nzero).
\end{equation}
In particular, for $\chi_0>0$, $0<\beta\leq1/3$, and $0<\varepsilon\leq Z_0\kappa_0\sqrt{\chi_0}/12$,
\begin{equation}\label{eq:fixed-scale-rate}
 K_L^*(\varepsilon,\beta;\Nzero^{(1)})
 =\Theta\!\left(\frac{Z_0^2\chi_0}{\varepsilon^2}\log\frac1\beta\right).
\end{equation}
\end{proposition}
\begin{proof}
Every normalized state $x=x_q$ from $\Nzero$ has positive finite $L^1$ norm and representatives
\[
 q^{(1)}(x)=\frac{Z_0x}{\norm{x}_1},\qquad q^{(2)}(x)=s_0x.
\]
Both are in $\Nzero$, and recover exactly the same state $x$ and target $F_\phi$. Therefore the three classes define the same setting, proving~\eqref{eq:fixed-scale-F}. On $\Nzero^{(1)}$, $L_\phi=Z_0F_\phi$, so rescaling an estimator and its error tolerance in both directions proves~\eqref{eq:fixed-scale-L}, so Theorem~\ref{thm:main} gives the rate.
\end{proof}

This result fixes one norm at a time, while fixing both might restrict the normalized states. For a known unit mass, $L_\phi=F_\phi$. Knowing only the $L^2$ norm instead gives $L_\phi(q)=s_0\int\phi x_q\,\mathrm d\nu$.

\subsection{Summaries that resolve target susceptibility}
A solver might keep conserved moments alongside the solution, or compute a few summaries once for use with several later observables. For example, Yamazaki et al.~\cite{yamazaki} extract low-order density Fourier modes and use them to calculate self-gravity. In this section, we study the hardness when the normalized state is supplied together with $m$ normalized linear summaries of the form:
\begin{equation}\label{eq:extra-moments}
 y_j(q)=F_{\psi_j}(q)
       =\frac{\int\psi_j q\,\mathrm d\nu}{\int q\,\mathrm d\nu},
\end{equation}
where the bounded real $\psi_j$, reference $q_0$, and neighborhood $\Nzero$ are known.

For fixed coefficients $c$,
\begin{equation}\label{eq:extra-residual}
 \begin{aligned}
 r_c&=\phi-\sum_jc_j\psi_j,\\
 F_\phi(q)&=\sum_jc_jy_j(q)+F_{r_c}(q).
 \end{aligned}
\end{equation}
Only the residual observable $r_c$ needs readout. Its best susceptibility at the reference state is:
\begin{equation}\label{eq:extra-chi}
 \chi_{\mathrm{rec},0}
   =\min_{c\in\R^m}\chi_{r_c}(q_0).
\end{equation}
Let $\ell_0=\int x_0\,\mathrm d\nu$ and $g_f=(f-F_f(q_0))/\ell_0$, which gives $\chi_f(q_0)=\norm{g_f}_2^2$. Therefore~\eqref{eq:extra-chi} projects the target sensitivity vector onto the span of summary sensitivity vectors, allowing the summaries to resolve that component, while the orthogonal residual governs readout cost~\cite{tsang}. This span is finite-dimensional and closed, so the minimum is attained.

\begin{proposition}[Readout with exact retained summaries]
\label{prop:extra-records}
Let $K^*_{\mathrm{rec}}$ be the smallest worst-case copy budget for one protocol that estimates $F_\phi(q)$ to error $\varepsilon$ with failure at most $\beta$ for every $q\in\Nzero$, receiving the corresponding summaries $y(q)$. If $\chi_{\mathrm{rec},0}>0$ and
\begin{equation}\label{eq:extra-range}
0<\varepsilon\leq\frac{\kappa_0\sqrt{\chi_{\mathrm{rec},0}}}{32},
 \qquad 0<\beta\leq1/3,
\end{equation}
then, with universal constants,
\begin{equation}\label{eq:extra-rate}
 K^*_{\mathrm{rec}}(\varepsilon,\beta;\Nzero)
 =\Theta\!\left(\frac{\chi_{\mathrm{rec},0}}{\varepsilon^2}
                      \log\frac1\beta\right).
\end{equation}
Two single-copy settings determined by known data and $q_0$ achieve the rate. The lower bound allows biased estimators and collective measurements. Exact zero-copy recovery throughout $\Nzero$ is possible precisely when $\phi=a+\sum_jc_j\psi_j$ almost everywhere for some constants $a,c_j$, equivalently $\chi_{\mathrm{rec},0}=0$.
\end{proposition}

For the upper bound, subtract the part of the target supplied by the summaries and estimate the residual. For the lower bound, we construct two positive inputs with exactly the same summaries but whose target values are $3\varepsilon$ apart, meaning that the summaries cannot distinguish them. Their fidelity divergence is $O(\varepsilon^2/\chi_{\mathrm{rec},0})$, so distinguishing their quantum states requires the stated number of copies. In Appendix~\ref{app:record-proof}  we give the construction and checks that both inputs remain in the neighborhood.

\paragraph{A linear-system example.}
On an unweighted grid, let $A$ be any known invertible real matrix, with positive solution $Aq=b$. Let its singular value decomposition be $A=U\Sigma V^{\mathsf T}$, so $Av_j=s_j u_j$ with $s_j>0$ with orthonormal bases $\{u_j\},\{v_j\}$. Assume the solver supplies
\begin{equation}\label{eq:matrix-mode-record}
 y_j(q)=F_{v_j}(q)=\frac{u_j^{\mathsf T}b}{s_j Z}.
\end{equation}
Producing these summaries requires $Z$. It is known in cases such as a unit-mass probability vector, but otherwise, acquiring $Z$ is a separate cost.

Each summary preserves one normalized solution component. If $P_k$ projects onto $\operatorname{span}\{g_{v_1},\ldots,g_{v_k}\}$, then
\begin{equation}\label{eq:modal-tail}
 \chi_{\mathrm{rec},0}^{(k)}
   =\left\|(I-P_k)g_\phi\right\|_2^2.
\end{equation}
If a new sensitivity vector has a nonzero component orthogonal to the earlier directions, normalize that component to $h$. The new summary then removes $|\ip{g_\phi}{h}|^2$ from the susceptibility, whereas a dependent direction removes nothing. Once the retained modes and the constant vector span $\phi$, exact recovery doesn't need more copies. In particular, retaining all components gives
\[
 F_\phi(q)=\sum_{j=1}^N(v_j^{\mathsf T}\phi)y_j(q),
 \qquad \chi_{\mathrm{rec},0}^{(N)}=0.
\]

The ordered modal-tail estimate and its copy-cost consequence are given in Appendix~\ref{app:dictionaries}.

\section{Readout from a reversible preparation circuit}\label{sec:queries}

Coherent access to the preparation routine improves the readout cost quadratically. For the same target and neighborhood, we show the optimal query cost is $\Theta(\sqrt{\chi_0}/\varepsilon)$. The susceptibility still governs the cost, while the precision dependence improves from $\varepsilon^{-2}$ to $\varepsilon^{-1}$. Amplitude estimation~\cite{bhmt} achieves this by estimating overlaps coherently.

Consider a finite-dimensional system with controlled access to $U_q,U_q^\dagger$, where $U_q\ket0=x_q$ has positive real amplitudes. Each controlled call to either unitary counts as one query. Known reference preparations and other computation are considered free in this model. The known center defining~\eqref{eq:neighborhood} is the quantity determining the overlap directions and required precisions.

We state the theorem uniformly over all unitary completions with first column $x_q$. This convention allows us to keep the rate related to the prepared state, even when unspecified columns might carry additional information.

For a full completion model, we let
\begin{equation}\label{eq:completion-class}
 \mathfrak U(q)=\{U:U\text{ unitary},\ U\ket0=\ket{x_q}\},
\end{equation}
and let one fixed $U\in\mathfrak U(q)$ be used consistently, together with its inverse. For a procedure $\mathcal A_T$ using at most $T$ calls to $U$ or $U^\dagger$, define its worst-case failure probability by
\begin{equation}\label{eq:coherent-risk}
 \begin{aligned}
 \mathcal{R}_\varepsilon(\mathcal A_T;\Nzero)
 :=\sup_{q\in\Nzero}\sup_{U\in\mathfrak U(q)}
 \PP_{q,U}\!\left(\abs{\mathcal A_T-F_\phi(q)}>\varepsilon\right).
 \end{aligned}
\end{equation}
Therefore
\begin{equation}
\label{eq:coherent-budget}
 \begin{aligned}
 T^*(\varepsilon,\beta;\Nzero)
 &=\min\{T:\exists\,\mathcal A_T\text{ with at most $T$ queries},\\[-2pt]
 &\hspace{8em}\mathcal R_\varepsilon(\mathcal A_T;\Nzero)\leq\beta\}.
 \end{aligned}
\end{equation}

We now state the main theorem of this section:
\begin{theorem}[Local coherent query complexity]\label{thm:queries}
For $\chi_0>0$ and $0<\varepsilon\leq\kappa_0\sqrt{\chi_0}/12$, the least deterministic query cap in~\eqref{eq:coherent-budget} achieving error at most $\varepsilon$ with success probability at least $2/3$, uniformly on $\Nzero$, satisfies
\begin{equation}\label{eq:query-rate}
 T^*(\varepsilon,1/3;\Nzero)
       =\Theta\!\left(\frac{\sqrt{\chi_0}}{\varepsilon}\right).
\end{equation}
For $0<\beta\leq1/3$, the same construction gives the upper bound
$O\!\left(\sqrt{\chi_0}\,\varepsilon^{-1}\log(1/\beta)\right)$.
\end{theorem}
We provide the proof for this result in Appendix~\ref{app:coherent-proof}.

\section{Applications: when the output interface matters}\label{sec:applications}

Linear systems, finite elements, and density evolution are some of the most basic tools of scientific computing. Implementing these algorithms efficiently in quantum computers often return the solution in amplitudes, while the required scientific answer needs an integral, average, or expected payoff.

For a linear system, the readout output is proportional to $q=A^{-1}b$~\cite{hhl}. If $q$ is a positive density, concentration, or allocation, normalized weighted sums are:
\[
 F_\phi(q)=\frac{\sum_i\phi_iq_i}{\sum_iq_i},\qquad q_i>0.
\]
These are precisely covered by Theorem~\ref{thm:main}, so their readout conditioning depends on the solution and the requested weights.

For finite elements, Montanaro--Pallister~\cite{montanaro} obtain a linear functional $\int r u_h$, which includes spatially averaged solutions. Our target is its normalized counterpart $\int r u_h/\int u_h$. For positive coefficients in a nonnegative basis, in Appendix~\ref{app:dictionaries} we give an exact translation from the coefficient state to $F_\phi(q)$. Recovering the unnormalized integral additionally requires its physical scale.

For heat and probability densities, the same integrals have a direct physical meaning. Linden--Montanaro--Shao estimate terminal heat $\int_Au(\cdot,T)$~\cite{linden}. For forward-density option pricing requires the payoff integral $\int g \,p_T$~\cite{guseynov}. Basis sampling of the solution state instead weights the squared density, so the susceptibility quantifies the statistical cost of recovering the linear functional.

Table~\ref{tab:interfaces} summarizes these output tasks. Published methods may retain scale, input data, or coherent preparation, which can provide specific workarounds to reduce the extraction cost in specific problems. For readout from state copies alone, Section~\ref{sec:modulus} specifies which bounds apply to the solver's possible outputs.

\begin{table}[t]
\centering\small
\begin{tabular}{@{}p{0.25\textwidth}p{0.29\textwidth}p{0.40\textwidth}@{}}
\toprule
Algorithmic setting & Positive field or coefficients & Normalized-linear quantity\\
\midrule
\raggedright Linear systems~\cite{hhl} & \raggedright $q=A^{-1}b>0$ & $\sum_i\phi_iq_i/\sum_iq_i$, including normalized sums and averages.\\[4pt]
\raggedright FEM~\cite{montanaro} & \raggedright $u_h=\sum_i c_i\psi_i$, $c_i>0$, $\psi_i\geq0$ & $\int r u_h/\int u_h$, with the exact translation from coefficient states given in Appendix~\ref{app:dictionaries}.\\[4pt]
\raggedright Heat and diffusion~\cite{linden} & \raggedright $q=u(\cdot,T)>0$, representing temperature or density & $\int_Au/\int u$, the regional heat or mass fraction.\\[4pt]
\raggedright Forward-density pricing~\cite{guseynov} & \raggedright Terminal density $q=p_T$ & $F_{g/B}(p_T)$ for $|g|\leq B$. Multiply by mass, payoff scale, and discount factor to obtain the price.\\[4pt]
\raggedright Collisionless Boltzmann~\cite{yamazaki} & \raggedright Phase-space distribution $q=f(x,v,T)$ & Normalized density Fourier modes, with bounded sine and cosine targets for real and imaginary parts.\\[4pt]
\raggedright Radiative transfer~\cite{igarashi} & \raggedright Directional intensities $q=(I_+,I_-)$ & Backward-stream energy fraction $F_{\one_{\{d=-1\}}}(q)$.\\[4pt]
\bottomrule
\end{tabular}
\caption{Uses of solution states produced by established quantum numerical
methods for normalized-linear quantities.}
\label{tab:interfaces}
\end{table}

The practical question we are concerned about is which solution profiles make these averages costly to recover. Our result in equation~\eqref{eq:chi} separates the effect of concentration from target sensitivity, showing that a uniform bound $Wq/Z\leq M$ gives $\chi\leq4M$, while a certain narrow component can make the cost grow. Consistent grid refinement preserves susceptibility by Proposition~\ref{prop:mesh}, so that it's not the grid size, but physical concentration that creates this penalty.

\subsection{Transport: a narrow pulse obscures a diffuse mass}\label{sec:transport}

In a lossless case of the two-stream update of Igarashi--Kadowaki--Kawabata~\cite{igarashi}, a narrow initial pulse separates into a concentrated component and a diffuse scattered field. Theorem~\ref{thm:transport} makes the readout cost concrete. We show that for the known pulse family with unknown scattering coefficient, recovering the backward-stream energy fraction from total-intensity copies costs $\Theta(\delta^{-1}\varepsilon^{-2}\log(1/\beta))$, where $\delta=m/T$ is the physical pulse width. The lower bound holds within this one-parameter family and allows arbitrary collective measurements. Supplying the mass, together with copies of the normalized diffuse component, reduces to $\chi=O(1)$ and removes the inverse-width factor. The model and proof are given in Appendix~\ref{app:applications}.

\section{Discussion and conclusions}\label{sec:discussion}

Our results develop an output-interface theory for normalized-linear readout. On the basis that preparing a normalized amplitude state is not the same task as producing a useful classical answer from that state, we develop concrete results for the case when the state carries $q$ in its amplitudes but the requested quantity is normalized linearly in $q$, and show that the interface has a target-dependent condition number. That condition number is $\chi_\phi(q)$, which accounts for both the concentration of the amplitude representation and the sensitivity of the requested observable.

The state-only interface is the central solved case. For a fixed positive reference state and its prescribed neighborhood, we show that the minimum number of state copies is
\[
 K^*(\varepsilon,\beta;\Nzero(q_0))
 =\Theta\!\left(\frac{\chi_0}{\varepsilon^{2}}\log\frac1\beta\right)
\]
for the adequate local accuracy and confidence. We show an upper bound, that two fixed single-copy measurements attain the same rate without a minimum-density assumption. Therefore $\chi_0$ quantifies the information cost of using this output representation.

The surrounding results show how this cost changes when the interface is enriched. Retained classical information remove target-sensitivity in the directions they provide, coherent preparation access improves the precision dependence quadratically. Adaptive allocation removes the need to know the reference when choosing the budget in exchange for a small extra cost, and the consistent discretization result shows that none of the factors are a discretization artifact.

We show how these bounds connects with states in common quantum algorithms, and provide concrete accounting for some special cases. Nevertheless, a lower bound for the state-only interface is not an end-to-end lower bound for every solver that produced the state, as the extra information available to the procedure is highly variable accross settings, here we show some of the most common such cases. In some settings, retaining useful classical information, circuit preparation, or targeting the observable directly define richer interfaces and can provide another route to the observable with lower cost. Our results guide solver design and modular computation by providing concrete guarantees for the readout cost of a copy-only output interface.

The relevant factor for algorithm-design is whether the outputs supplied to the readout stage contain hard alternatives at the scale of the susceptibility $\chi$. State-preparation complexity should in general be reported together with the readout interface it supports. A natural next step is to design classical summaries and quantum output states jointly for a specified family of observables, while accounting for both their production and readout costs.

\section*{Author contributions}
Roi Vence-Nogueira developed and identified this setting present in many relevant numerical algorithms, developed the ideas, created the proofs and wrote the manuscript in full. LLMs helped locate some of the papers on the quantum numerical methods discussed in the applications section and appendix, and assisted with converting the LaTeX source to the journal format.

\newpage

\appendix
\section{Proofs of the main results}\label{app:main-proof}

This appendix proves the main theorem and the retained-summary extension, and Appendix~\ref{app:infinite} records why the finite-measure assumption is needed. The proof of the main theorem has three steps: a pointwise ratio estimate, uniform control over the reference neighborhood, and a positive pair of hard alternatives for the converse.

\subsection{Proof of Proposition~\ref{prop:chi}: Explicit susceptibility}
\begin{proof}
Direct integration gives $\E_Q[w]=0$ and $\Cov_P(\phi,h)=\E_Q\left[w\left(h-\E_Q[h]\right)\right]$. Using Cauchy--Schwarz inequality we can bound the quotient by $\E_Q[w^2]$, which is finite by finite measure and boundedness of $\phi$. 

The clipped $w\,\mathrm{1}_{[-M,M]}\overset{M\to\infty}{\to} w$ in $L^2(Q)$ , and substituting these bounded functions shows that the quotient converges to $\E_Q[w^2]$ when it is positive. Cancelling $q^2$ gives the first closed form, and centering around $\E_U[\phi]$ gives the second. Using positivity proves the zero case. Finally, differentiating the two normalizations proves~\eqref{eq:derivatives}.
\end{proof}

\subsection{Proof of Theorem~\ref{thm:main}: Local minimax copy complexity}
\subsubsection{Achievability with two fixed measurements}
Using the uniform reference $v$ from Section~\ref{sec:model}, define two fixed measurements depending only on $\nu$ and $\phi$:
\begin{equation}\label{eq:public}
 \begin{aligned}
 a&=\E_U[\phi], & z&=(\Phi-aI)v,\\
 \sigma^2&=\norm{z}^2=\Var_U(\phi),\\
 C&=\frac{\ket v\bra z+\ket z\bra v}{2},\\
 D&=\ket v\bra v.
 \end{aligned}
\end{equation}
The measurement $D$ projects onto $v$, while $C$ couples $v$ to the centered target $z$. The outcomes of $D$ are $0$ and $1$, while those of $C$ are $0$ and $\pm\sigma/2$. If $\sigma=0$, the answer is $a$. Otherwise
\[
 \mu=a+\frac{\ip{x_q}{Cx_q}}{\ip{x_q}{Dx_q}},
 \qquad \ip{x_q}{Dx_q}=\kappa^2.
\]
Measure each setting on $n$ fresh copies and denote the independent sample means by $\overline C,\overline D$. The estimator is
\begin{equation}\label{eq:estimator}
 T_n=\begin{cases}
 \clip_{[-1,1]}(a+\overline C/\overline D),&\overline D>0,\\
 a,&\overline D=0.
 \end{cases}
\end{equation}
Both operators have rank at most two and are well defined on $L^2(\nu)$. Therefore the weak-value ratio is realized using only strong measurements~\cite{cohen-pollak}.

\begin{lemma}[Pointwise guarantee]\label{lem:public}
For every positive input, every $n\geq1$, and every $\varepsilon>0$,
\begin{equation}\label{eq:pointwise-failure}
 \PP_q\left[\abs{T_n-\mu}>\varepsilon\right]
 \leq\frac4{n\kappa^2}+\frac{5\chi}{n\varepsilon^2}.
\end{equation}
Thus $n\geq27\chi/\varepsilon^2$ suffices for failure at most $1/3$ whenever $\varepsilon\leq\kappa\sqrt\chi$.
\end{lemma}
\begin{proof}
We control two possible failures separately. The denominator estimate may be too small, and the numerator--denominator combination may fluctuate. Put $\Delta=\mu-a$. The two means are
\[
 N:=\E[C]=\kappa^2\Delta,
 \qquad M:=\E[D]=\kappa^2.
\]
The centered vector $z$ is orthogonal to $v$, so
\[
 \begin{aligned}
 \E[C^2]=\norm{Cx}^2
 &=\frac{\kappa^2\sigma^2+|\ip{z}{x}|^2}{4}
 =\frac{\kappa^4\chi}{4},\\
 \Var(\overline D)&\leq\frac{\kappa^2}{n}.
 \end{aligned}
\]
Let $G=\{\overline D\geq M/2\}$. Chebyshev's inequality gives
\[
 \PP(G^c)
 \leq\frac{4\Var(\overline D)}{M^2}
 \leq\frac{4}{n\kappa^2}.
\]
It remains to control the ratio on $G$. Define
\[
 V=(\overline C-N)-\Delta(\overline D-M).
\]
The two sample means are independent and $N=\Delta M$, hence $\E[V]=0$ and
\[
 \E\!\left[V^2\right]\leq\frac{\kappa^4\chi}{4n}
       +\frac{\kappa^2\Delta^2}{n}
       \leq\frac{5\kappa^4\chi}{4n}.
\]
On $G$, the untruncated ratio has error
\[
 a+\frac{\overline C}{\overline D}-\mu
 =\frac{\overline C-\Delta\overline D}{\overline D}
 =\frac{V}{\overline D}.
\]
Since $\overline D\geq M/2$, the squared error is at most $4V^2/M^2$. Clipping cannot increase the distance to $\mu\in[-1,1]$, so
\[
 \E\left[(T_n-\mu)^2\one_G\right]
 \leq\frac{4\E[V^2]}{M^2}
 \leq\frac{5\chi}{n}.
\]
Markov's inequality on this event, followed by the bound for $G^c$, proves~\eqref{eq:pointwise-failure}.
\end{proof}

\subsubsection{Uniformity over the neighborhood}
The pointwise lemma uses the actual state's $\kappa$ and $\chi$. We now replace both by bounds depending only on the reference state. The measurements remain fixed throughout the neighborhood. Set
\begin{equation}\label{eq:rd}
 r=\frac v{\kappa_0},\qquad
 d=(\Phi-\mu_0 I)r.
\end{equation}
Then $\ip{r}{x_0}=1$, $\ip{d}{x_0}=0$, $\norm{r}=1/\kappa_0$, and $\norm{d}^2=\chi_0$. For $x=x_q\in\Nzero$, the neighborhood condition gives $|\ip{r}{x-x_0}|\leq1/4$ and $|\ip{d}{x-x_0}|\leq\kappa_0\sqrt{\chi_0}/4$. Therefore
\begin{equation}\label{eq:local-controls}
 \begin{aligned}
 \frac34\leq\ip{r}{x}\leq\frac54,\qquad
 \abs{\ip{d}{x}}\leq\frac{\kappa_0\sqrt{\chi_0}}4,\\
 \mathcal F(x)-\mu_0=\frac{\ip{d}{x}}{\ip{r}{x}}.
 \end{aligned}
\end{equation}
The first two bounds imply $\kappa\geq3\kappa_0/4$ and
$|\mu-\mu_0|\leq\kappa_0\sqrt{\chi_0}/3$. To compare the susceptibilities, use the norm of the centered vector at $x$ and the triangle inequality:
\begin{equation}\label{eq:chi-comparison}
 \sqrt\chi=\frac{\norm{(\Phi-\mu I)v}}\kappa
 \leq\frac{\kappa_0\sqrt{\chi_0}+\abs{\mu-\mu_0}}
              {3\kappa_0/4}
 \leq\frac{16}{9}\sqrt{\chi_0}.
\end{equation}
Substituting these two uniform bounds into Lemma~\ref{lem:public} gives
\[
 \PP_q\left[\abs{T_n-\mu}>\varepsilon\right]
 \leq\left(\frac4{81}+\frac{1280}{81}\right)
             \frac{\chi_0}{n\varepsilon^2}
 =\frac{428}{27}\frac{\chi_0}{n\varepsilon^2}.
\]
Therefore $n=\lceil48\chi_0/\varepsilon^2\rceil$ makes one block fail with probability below $1/3$, uniformly over the neighborhood. For smaller $\beta$, take an odd number $m(\beta)$ of independent blocks and return their median. If $B_j$ denotes the failure indicator of block $j$, the median fails only when $\sum_jB_j\geq m(\beta)/2$. The exponential Markov bound applied to $2^{\sum_jB_j}$ gives
\[
 \PP\left[\text{median fails}\right]
 \leq (8/9)^{m(\beta)/2}
 \leq e^{-m(\beta)/18}.
\]
The stated choice of $m(\beta)$ makes this at most $\beta$, proving the upper bound in~\eqref{eq:main-bounds}.

\subsubsection{Optimality: a positive hard pair}
The lower bound needs two allowed positive states whose target values differ but whose amplitude states are difficult to distinguish. The most sensitive direction $d$ need not be nonnegative, so we split it into $d=d_+-d_-$. The two parts have disjoint supports and satisfy
\[
\begin{aligned}
 c&=\ip{x_0}{d_+}=\ip{x_0}{d_-}\geq0,\\
 A_+&=\norm{d_+}^2,\qquad A_-=\norm{d_-}^2,\qquad
 A_++A_-=\chi_0.
\end{aligned}
\]
Choose $t=3\varepsilon/\chi_0$ and define
\begin{equation}\label{eq:positive-pair}
 y_\pm=x_0+td_\pm,\qquad
 x_\pm=\frac{y_\pm}{\norm{y_\pm}}.
\end{equation}
Because $x_0>0$ and $d_\pm\geq0$, both $y_\pm$ are positive. Their norms are at least one, and normalization is a contraction toward the unit sphere, so
\[
 \norm{x_\pm-x_0}\leq t\norm{d_\pm}\leq\frac{\kappa_0}{4}.
\]
Thus both alternatives lie in $\Nzero$.

Scale invariance and~\eqref{eq:rd} give
\[
\begin{aligned}
 \mathcal F(x_+)-\mu_0&=\frac{tA_+}{1+t\ip{r}{d_+}},\\
 \mathcal F(x_-)-\mu_0&=-\frac{tA_-}{1+t\ip{r}{d_-}}.
\end{aligned}
\]
Each denominator lies in $[1,5/4]$, so the target gap is at least $(4/5)t\chi_0=12\varepsilon/5>2\varepsilon$. To quantify how close the two states remain, put $B=1+2tc\geq1$. Their squared fidelity is
\begin{equation}\label{eq:pair-fidelity}
 \begin{aligned}
 f=\abs{\ip{x_+}{x_-}}^2
   &=\frac{B^2}{(B+t^2A_+)(B+t^2A_-)},\\
 -\log f&\leq t^2\chi_0=\frac{9\varepsilon^2}{\chi_0}.
 \end{aligned}
\end{equation}
An estimator with error at most $\varepsilon$ on both alternatives distinguishes them by thresholding at the midpoint of their target values. Any protocol using at most $k$ copies, including an adaptive protocol or one that discards copies, is a measurement on the padded $k$-copy state. The two tensor-power states have trace distance $\sqrt{1-f^k}$. Helstrom's bound~\cite{helstrom} therefore gives
\[
 1-2\beta\leq\sqrt{1-f^k},
\]
which is equivalent to
\begin{equation}\label{eq:helstrom}
 k\geq\frac{\log(1/[4\beta(1-\beta)])}{-\log f}.
\end{equation}
Equations~\eqref{eq:pair-fidelity}--\eqref{eq:helstrom} prove the lower bound. For $0<\beta\leq1/3$, the logarithm above is comparable to $\log(1/\beta)$. The inequality $\chi_0/\varepsilon^2\geq144$ absorbs the ceiling in the upper bound. This completes Theorem~\ref{thm:main}.

\subsection{Proof of Proposition~\ref{prop:basis-gap}: Basis-only separation}\label{app:basis-separation}

This subsection proves Proposition~\ref{prop:basis-gap}. The construction compares two mixtures of basis-sampling transcripts. We partition $A$ and $A^c$ into coordinate pairs and hide the perturbation behind an independent random sign in each pair~\cite{paninski}. A single observation can reveal at most one coordinate from a pair, so the hidden sign becomes informative only when a pair is observed twice.

For counting measure, basis readout estimates
\[
 F_\phi(q)=\frac{\sum_i\phi_i\sqrt{Q_i}}{\sum_i\sqrt{Q_i}}.
\]
The denominator is the square-root power sum underlying R\'enyi entropy of order $1/2$~\cite[Sec.~I C]{acharya-renyi}. Proposition~\ref{prop:basis-gap} concerns a weighted ratio on a near-uniform subclass. Unrestricted entropy estimation rates do not directly give its sample complexity.

In hypothesis $H_A$, assign an independent uniformly random sign $s\in\{-1,1\}$ to each pair in $A$, and give its two coordinates Born probabilities
\[
 \left(\frac{1+s/2}{D},\frac{1-s/2}{D}\right).
\]
All coordinates in $A^c$ have probability $1/D$. Hypothesis $H_B$ exchanges the two halves. Each sign realization defines a positive state $x=\sqrt Q$ and therefore an admissible amplitude input $q=x$. Writing $r=(\sqrt{3/2}+\sqrt{1/2})/2$, every such state satisfies
\[
 \begin{aligned}
 \ip{x}{x_0}&=\frac{1+r}{2},
 & \norm{x-x_0}^2&=1-r<\frac1{16},\\
 F(H_A)&=\frac r{1+r},
 & F(H_B)&=\frac1{1+r}.
 \end{aligned}
\]
Therefore all inputs belong to $\Nzero$, with target gap $(1-r)/(1+r)>1/64=2\varepsilon_*$ between sign families. Also $\varepsilon_*<1/24=\kappa_0\sqrt{\chi_0}/12$.

Let $L_A,L_B$ be the laws of $k$ basis outcomes after averaging over the signs. Every pair has total Born mass $2/D$, so the pair labels are independent and uniform over $D/2$ possibilities. Let $E$ be the event that no pair is observed twice. Then
\[
 \PP(E^c)\leq\binom{k}{2}\frac{2}{D}.
\]
Consider any ordered sequence of $k$ coordinate outcomes in $E$. Each outcome comes from a different pair. Averaging over the hidden signs gives probability $1/D$ for each coordinate under either hypothesis. Different pairs have independent signs, so the probability of the whole sequence is $(1/D)^k=D^{-k}$ under both hypotheses. $L_A$ and $L_B$ assign the same probability to every sequence without a repeated pair, so they can differ only on sequences in $E^c$. Then,
\[
 \operatorname{TV}(L_A,L_B)\leq\PP(E^c)
 \leq\binom k2\frac2D=\frac{k(k-1)}D.
\]
If an estimator was accurate for both mixtures, thresholding its output at $1/2$ would give a test with average error at most $1/3$. But a test like that requires total variation at least $1/3$, so $k(k-1)/D\geq1/3$. This proves~\eqref{eq:basis-gap}.  Theorem~\ref{thm:main} with $\chi_0=1/4$ gives the dimension-independent upper bound for the two fixed quantum measurements.

\paragraph{Continuum consequence.}
For Lebesgue measure on $[0,1]$, $q_0=1$, and $\phi=\one_{[0,1/2]}$, no finite deterministic position-sample cap works uniformly on $\Nzero$ at $(\varepsilon_*,1/3)$, even under $\sqrt{1/2}\leq q\leq\sqrt{3/2}$. The quantum copy cap remains finite. Indeed, for $D$ divisible by four, partition $[0,1]$ into equal intervals $I_i$ and set $q(t)=\sqrt{DQ_i}$ on $I_i$. Its $L^2$ norm is one, and the bounds, distance, and target calculations above are unchanged. A position sample supplies the cell label and a parameter-independent uniform position inside that cell, so the same testing bound applies. For any finite cap $k$, choosing $D>3k(k-1)$ contradicts uniform success on this fixed continuum neighborhood.

\subsection{Proof of Theorem~\ref{thm:adaptive}: Adaptive allocation}\label{app:pilot}

The adaptive protocol has two tasks. It must estimate the rare-event probabilities that determine the budget, and it must show that the resulting random budget has controlled expectation and tails. We handle these tasks in that order.

For $\sigma>0$, run two independent pilots. The first stops after $r$ observations with $D=1$, and the second stops after $r$ nonzero outcomes of $C$, where
\[
 r=\max\{3,\lceil8\log(8/\beta)\rceil\}.
\]
If their stopping times are $N_D$ and $N_P$, define
\begin{equation}\label{eq:pilot-estimates}
 \begin{aligned}
 \widehat\eta&=r/N_D,& \widehat p&=r/N_P,\\
 \eta_L&=\widehat\eta/2,& p_U&=\min\{1,2\widehat p\},\\
 \chi_U&=\sigma^2p_U/\eta_L^2.
 \end{aligned}
\end{equation}
Use $n=\lceil12/\eta_L+15\chi_U/\varepsilon^2\rceil$ fresh copies per setting in each of
\[
 m_{\rm ad}(\beta)=2\left\lceil9\log(2/\beta)-\tfrac12\right\rceil+1
\]
independent blocks. Return the median of the estimates in \eqref{eq:estimator}. The total count is $K=N_D+N_P+2m_{\rm ad}(\beta)n$. For $\sigma=0$, return the known constant target without copies.

We will use the following inverse-binomial estimates. Let $N$ be the number of independent Bernoulli($u$) trials needed to see $r\geq3$ successes, with $0<u\leq1$, and set $\widehat u=r/N$. Then
\begin{equation}\label{eq:inverse-bounds}
 \begin{aligned}
 \PP\left[\widehat u\notin[u/2,2u]\right]&\leq2e^{-r/8},\\
 \E\left[N\right]&=r/u,\\
 \E\left[N^2\right]&=\frac{r^2+r(1-u)}{u^2},\\
 \E\left[\frac rN\right]&\leq\frac{ru}{r-1}.
 \end{aligned}
\end{equation}
For the upper tail, let $S_j$ be the number of successes in the first $j$ trials. The event $\widehat u>2u$ implies $N<r/(2u)$, hence $S_j\geq r$ at $j=\lceil r/(2u)\rceil-1$. An exponential Markov bound for $S_j$ gives $e^{-r/8}$. 

For the lower tail, $\widehat u<u/2$ implies $N>2r/u$, hence $S_j<r$ at $j=\lfloor2r/u\rfloor$. Applying the corresponding lower-tail bound, equivalently to $2^{-S_j}$, gives another $e^{-r/8}$. This proves the first line of~\eqref{eq:inverse-bounds}.

The first two moments follow by writing $N$ as a sum of $r$ geometric waiting times. For the last moment, when $0<u<1$ use the negative-binomial mass
\[
 \PP(N=j)=\binom{j-1}{r-1}u^r(1-u)^{j-r},
 \qquad j\geq r.
\]
The identity
\[
 \frac{r-1}{j-1}\binom{j-1}{r-1}=\binom{j-2}{r-2}
\]
gives $\E[(r-1)/(N-1)]=u$. Since $r/N\leq r/(N-1)$, the final inequality follows. The case $u=1$ is immediate.

For Theorem~\ref{thm:adaptive}, let $H$ be the event that both pilots are within a factor two. Equation~\eqref{eq:inverse-bounds} gives $\PP(H^c)\leq4e^{-r/8}\leq\beta/2$. On $H$,
\[
 \eta/4\leq\eta_L\leq\eta,\quad p\leq p_U\leq4p,\quad
 \chi\leq\chi_U\leq64\chi,
\]
Therefore,
\begin{equation}\label{eq:explicit-tail}
 K\leq B_q:=\frac{4r}{\kappa^2}
    +2m_{\rm ad}(\beta)\left(1+\frac{48}{\kappa^2}
                  +\frac{960\chi}{\varepsilon^2}\right).
\end{equation}
Condition on the pilot data and on $H$. The fresh blocks then use independent copies, and the conservative bounds above make each block fail with probability at most $1/3$. The median only fails when at least half of the block indicators are one. The same exponential Markov estimate used in the main proof gives here a conditional failure probability of at most $e^{-m_{\rm ad}(\beta)/18}\leq\beta/2$. Combining this with $\PP(H^c)\leq\beta/2$ gives
\[
 \begin{aligned}
 \PP\left[K\leq B_q\right]&\geq1-\beta/2,\\
 \PP\left[K\leq B_q,\abs{\widehat F-\mu}\leq\varepsilon\right]&\geq1-\beta.
 \end{aligned}
\]

For the expectation, $\E[N_D]=r/\eta$, $\E[N_P]=r/p\leq r/\eta$, $\E[N_D^2]\leq r(r+1)/\eta^2$, and $\E[1/\eta_L]=2/\eta$. The two pilots are independent and $p_U\leq2r/N_P$, so
\[
 \begin{aligned}
 \E\left[\chi_U\right]&=\frac{4\sigma^2}{r^2}\E[N_D^2p_U]\\
 &\leq\frac{8\sigma^2}{r^2}\frac{r(r+1)}{\eta^2}\frac{rp}{r-1}\\
 &=8\frac{r+1}{r-1}\chi\leq16\chi.
 \end{aligned}
\]
Therefore
\begin{equation}\label{eq:explicit-expectation}
 \E[K]\leq\frac{2r}{\kappa^2}
       +2m_{\rm ad}(\beta)\left(1+\frac{24}{\kappa^2}
                    +\frac{240\chi}{\varepsilon^2}\right),
\end{equation}
proving the cost bounds.

\subsection{Proof of Proposition~\ref{prop:mesh}: Consistent discretization}\label{app:mesh-proof}
\begin{proof}
Finite measure converts $L^2$ convergence to $L^1$ convergence, and squared $L^2$ norms converge. For the numerator,
\[
\begin{aligned}
 \left|\int q_h\phi_h\,\mathrm d\nu-\int q\phi\,\mathrm d\nu\right|
 &\leq\norm{q_h}_2\norm{\phi_h-\phi}_2\\
 &\quad+\norm{q_h-q}_2\norm\phi_2\longrightarrow0.
\end{aligned}
\]
Since $Z>0$, the ratios converge. Also $\phi_h-\mu_h\to\phi-\mu$ in $L^2$, proving susceptibility convergence. The isometry $\ket i\mapsto\one_{A_{h,i}}/\sqrt{\nu_{h,i}}$ sends the discrete weighted amplitude vector to $q_h/\norm{q_h}_2$. The discrete experiment is thus exactly the experiment on these piecewise-constant functions.
\end{proof}

\subsection{Proof of Proposition~\ref{prop:extra-records}: Retained summaries}\label{app:record-proof}

We prove Proposition~\ref{prop:extra-records} by applying the main theorem to the residual observable for the upper bound and constructing alternatives with exactly the same retained summaries for the lower bound.
\begin{proof}
The proof has two independent parts. For the upper bound, the summaries remove a known part of the target and Theorem~\ref{thm:main} is applied to the residual. For the lower bound, we construct two positive inputs with the same summaries but separated target values.

For the upper bound, let $c_*$ minimize~\eqref{eq:extra-chi} and write $r_*=r_{c_*}$. Set $B=\max(1,\norm{r_*}_\infty)$. The residual $r_*/B$ is bounded, and the identity
\[
 F_\phi(q)=\sum_jc_{*,j}y_j(q)+F_{r_*}(q)
\]
means that the known summaries supply the first term exactly. Apply Theorem~\ref{thm:main} to $r_*/B$ with tolerance $\varepsilon/B$. Multiplying the resulting estimate by $B$ restores the desired tolerance. Rescaling changes neither the required number of copies nor the relevant accuracy condition.

For the lower bound, every summary must take the same value on the two alternatives. Put $p_0=x_0/\ell_0$ and $d=g_{r_*}$, where $g_f=(f-F_f(q_0))/\ell_0$ is the centered sensitivity vector defined in the main text. Orthogonality of the projection gives
\[
 \begin{aligned}
 \norm{d}_2^2&=\chi_{\mathrm{rec},0}, & \ip{x_0}{d}&=0,\\
 \ip{g_{\psi_j}}{d}&=0, & \ip{g_\phi}{d}&=\chi_{\mathrm{rec},0}.
 \end{aligned}
\]
Therefore $d$ is invisible to every retained summary but changes the target. We use its positive and negative parts and add a mass correction so that both perturbed functions remain normalized. Define
\begin{equation}\label{eq:extra-pair}
 \begin{aligned}
 A_\pm&=\int d_\pm\,\mathrm d\nu,&
 u&=\frac{3\varepsilon}{\ell_0\chi_{\mathrm{rec},0}},\\
 p_\pm&=(1-uA_\pm)p_0+u d_\pm.
 \end{aligned}
\end{equation}
The correction makes both $p_+$ and $p_-$ have unit mass. Their difference is
\[
 p_+-p_-=u\bigl(d-(A_+-A_-)p_0\bigr).
\]
Using $\int d\,\mathrm d\nu=A_+-A_-$, direct substitution gives, for any bounded $f$,
\[
 F_f(p_+)-F_f(p_-)=u\ell_0\ip{g_f}{d}.
\]
The orthogonality relations make all retained summaries equal, while the target values differ by $3\varepsilon$. The accuracy assumption implies $uA_\pm\leq3/32$, so both densities remain positive. Their normalized amplitude rays have the form $x_0+t_\pm d_\pm$, with
\[
 t_\pm=\frac{u\ell_0}{1-uA_\pm}
 \leq\frac{96\varepsilon}{29\chi_{\mathrm{rec},0}}.
\]
Although the two values of $t_\pm$ don't necessarily coincide, the resulting rays are still close. Since
$\ip{x_0}{d}=0$, while $d=d_+-d_-$ and $d_+,d_-$ have disjoint supports,
$\ip{x_0}{d_+}=\ip{x_0}{d_-}$ and $\ip{d_+}{d_-}=0$. Radial projection onto the unit ball is nonexpansive, so
with $t_{\max}=\max(t_+,t_-)$,
\[
 \begin{aligned}
 \norm{x_{p_\pm}-x_0}_2&\leq t_\pm\norm{d_\pm}_2,\\
 \norm{x_{p_+}-x_{p_-}}_2
 &\leq t_{\max}\bigl(\norm{d_+}_2+\norm{d_-}_2\bigr)\\
 &\leq \frac{96\sqrt{2}\,\varepsilon}
              {29\sqrt{\chi_{\mathrm{rec},0}}}.
 \end{aligned}
\]
Consequently, for $f=\abs{\ip{x_{p_+}}{x_{p_-}}}^2$,
\[
 1-f\leq\norm{x_{p_+}-x_{p_-}}_2^2
 \leq 2\left(\frac{96}{29}\right)^2
       \frac{\varepsilon^2}{\chi_{\mathrm{rec},0}}<22\frac{\varepsilon^2}{\chi_{\mathrm{rec},0}}.
\]
The accuracy range implies $\varepsilon^2/\chi_{\mathrm{rec},0}\leq1/1024$, so
$-\log f\leq(1-f)/f<25\varepsilon^2/\chi_{\mathrm{rec},0}$. Thus
\[
 \begin{aligned}
 \norm{x_{p_\pm}-x_0}_2&\leq3\kappa_0/29<\kappa_0/4,\\
 -\log\abs{\ip{x_{p_+}}{x_{p_-}}}^2
 &<25\varepsilon^2/\chi_{\mathrm{rec},0}.
 \end{aligned}
\]
Therefore both alternatives lie in $\Nzero$, and~\eqref{eq:helstrom} gives the lower bound. The common summaries need not equal $y(q_0)$, because the guarantee varies over the whole neighborhood. If the residual is constant, the target is an affine function of the retained summaries, which proves the zero-copy sufficiency. Conversely, if $\chi_{\mathrm{rec},0}>0$, the same equal-summary construction is possible for every sufficiently small positive $\varepsilon$. The target therefore differs accross states that have identical summaries, so exact zero-copy recovery throughout $\Nzero$ is impossible.
\end{proof}

\subsection{Proof of Theorem~\ref{thm:queries}: Coherent query complexity}\label{app:coherent-proof}

\begin{proof}
We prove the two directions separately. The upper bound estimates the two overlaps that determine the target ratio. The lower bound shows that nearby preparation unitaries remain indistinguishable to any short-query procedure.

For the upper bound, let $d$ be the centered sensitivity direction from~\eqref{eq:rd} and set
$e=d/\sqrt{\chi_0}$. Write
\[
 c=\ip{v}{x},\qquad z'=\ip{e}{x}.
\]
For every $x=x_q$ in $\Nzero$, the neighborhood estimates give
\[
 c\geq\frac{3\kappa_0}{4},\qquad
 |z'|\leq\frac{\kappa_0}{4},
\]
and the target has the exact representation
\[
 \mathcal F(x)-\mu_0
 =\kappa_0\sqrt{\chi_0}\,\frac{z'}{c}.
\]

Hadamard interference converts a real overlap $s$ into the probability $(1+s)/2$. Amplitude estimation~\cite[Theorem~12]{bhmt} therefore estimates both $c$ and $z'$ to additive accuracy
\[
 \delta=\frac{\varepsilon}{16\sqrt{\chi_0}}
\]
with failure probability at most $1/6$ each, using $O(1/\delta)$ controlled preparation and inverse queries. Denote the estimates by $\widehat c$ and $\widehat z'$, set
\[
 \widetilde c=\max\!\left\{\frac{\kappa_0}{2},\widehat c\right\},
\]
and output
\[
 \mu_0+\kappa_0\sqrt{\chi_0}\,\frac{\widehat z'}{\widetilde c}.
\]
On the event that both overlap estimates are accurate, the accuracy assumption gives $\delta\leq\kappa_0/192$, so $\widehat c\geq3\kappa_0/4-\delta>\kappa_0/2$. The ratio error is then bounded by
\[
 \left|\frac{\widehat z'}{\widetilde c}-\frac{z'}c\right|
 \leq \frac{2}{\kappa_0}\left(\delta+\frac{|z'|}{c}\delta\right).
\]
Using $|z'|/c\leq1/3$ gives
\[
 \left|\widehat{\mathcal F}-\mathcal F(x)\right|
 \leq 2\sqrt{\chi_0}\left(\delta+\frac{|z'|}{c}\delta\right)
 \leq \frac{8}{3}\sqrt{\chi_0}\,\delta
 <\varepsilon.
\]
Thus a constant number of amplitude-estimation runs gives the required constant-confidence estimate with $O(\sqrt{\chi_0}/\varepsilon)$ queries. Independent repetition followed by a median increases the success probability to $1-\beta$ and multiplies the query count by $O(\log(1/\beta))$.

For the lower bound, use the positive alternatives $x_+$ and $x_-$ from~\eqref{eq:positive-pair}. They have target values separated by more than $2\varepsilon$. Choose a known unitary $U_0$ with $U_0\ket 0=x_0$. For each sign, let $R_\pm$ be the real rotation in the plane spanned by $x_0$ and $x_\pm$ that maps $x_0$ to $x_\pm$ and acts as the identity on the orthogonal complement. Define
\[
 U_\pm=R_\pm U_0.
\]
The first columns of $U_+$ and $U_-$ are the two hard states. A plane rotation fulfills
$\|R_\pm-I\|_{\rm op}=\|x_\pm-x_0\|$, so:
\[
\begin{aligned}
 \|U_+-U_-\|_{\rm op}
 &\leq \|x_+-x_0\|+\|x_--x_0\| \\
 &\leq t\left(\|d_+\|+\|d_-\|\right)
 \leq \frac{3\sqrt2\,\varepsilon}{\sqrt{\chi_0}}.
\end{aligned}
\]
The same bound holds for controlled queries and inverse queries. A standard telescoping hybrid argument~\cite[Section~IV~A]{montanaro} bounds the trace distance between the final states of any $T$-query procedure by
\[
 T\,\|U_+-U_-\|_{\rm op}.
\]
Measurements and auxiliary computations can be put at the end, so that this applies to arbitrary coherent query algorithms.

If the procedure succeeds with probability at least $1-\beta$ on both alternatives, thresholding its estimate at the midpoint of the two target values gives a binary test with error at most $\beta$. Therefore, the trace distance must be at least $1-2\beta$. Then:
\[
 T\geq
 \frac{(1-2\beta)\sqrt{\chi_0}}{3\sqrt2\,\varepsilon}.
\]
At $\beta=1/3$ this matches the upper bound.
\end{proof}

\section{Extensions and supporting material}\label{app:supporting}

This appendix gives one restricted-family readout theorem, translates the
susceptibility into several standard numerical representations, and records
the corresponding ratio result for signed and complex amplitudes. The
transport proof exhibits admissible hard alternatives generated by a physical
update rule. The dictionaries are interpretive applications of the main
theorem, not additional lower bounds.

\subsection{Transport application}\label{app:applications}

We specialize the two-stream radiative-transfer update of Igarashi--\allowbreak Kadowaki--\allowbreak Kawabata~\cite[Eqs.~(3) and (8)]{igarashi}. Set the absorption and source terms to zero, take the scattering coefficient to be $2\lambda$, and set the mesh and time steps to $1/T$, in units with stream speed one. The collision-and-streaming update is
\begin{equation}\label{eq:transport-update}
 \begin{aligned}
 q^{n+1}_d(j+d)&=(1-\lambda/T)q^n_d(j)\\
 &\quad+(\lambda/T)q^n_{-d}(j),\qquad d\in\{-1,1\}.
 \end{aligned}
\end{equation}
For integers $T\geq4$, $1\leq m\leq T/4$, initialize a uniform, unit-mass pulse of width $m$ moving to the right on a periodic grid with $M\geq4T+m$ cells and let $q_\lambda$ be the field after $T$ steps, for $\lambda\in[1/4,1/2]$. We use counting measure on the known reachable support, which has $W=\Theta(T)$ labels. The field $q_\lambda$ is positive on this support for every $\lambda$ in the interval. The physical pulse width is $\delta=m/T$.

The target is the backward-stream energy fraction, with $\phi(j,d)=\one_{\{d=-1\}}$. Mass is conserved, and summing the two-stream recurrence gives
\begin{equation}\label{eq:transport-target}
 \mu_T(\lambda)=\sum_jq_{\lambda,-}(j)
       =\frac{1-(1-2\lambda/T)^T}{2}.
\end{equation}
\begin{theorem}[Pulse readout]\label{thm:transport}
Let $\mathcal C_{T,m}=\{q_\lambda:1/4\leq\lambda\leq1/2\}$ in \eqref{eq:transport-update}, and let $x_\lambda=q_\lambda/\norm{q_\lambda}_2$. Uniformly over this family,
\begin{equation}\label{eq:transport-chi}
 \begin{aligned}
 \kappa^{-2}&=\Theta(T/m),\\
 \chi_\phi(q_\lambda)&=\Theta(T/m),\\
 \frac{\mu_T'(\lambda)^2}{\norm{x'_\lambda}_2^2}
 &=\Theta(T/m).
 \end{aligned}
\end{equation}
For $0<\varepsilon\leq1/48$ and $0<\beta\leq1/3$,
\begin{equation}\label{eq:transport-rate}
 K^*(\varepsilon,\beta;\mathcal C_{T,m})
 =\Theta\!\left(\frac{T}{m\varepsilon^2}\log\frac1\beta\right).
\end{equation}
The known data are $T,m,M$, the initial pulse, the coefficient interval, and the target.
\end{theorem}

At fixed physical width, mesh refinement does not change this order. The component that doesn't change has mass $a_\lambda=(1-\lambda/T)^T=\Theta(1)$, whereas the conditionally normalized changed-path state $r_\lambda$ has $\chi_\phi(r_\lambda)=O(1)$. If the interface supplies $a_\lambda$ and copies of this separately normalized diffuse state, the factor $T/m$ disappears.In this family, exact $a_\lambda$ alone determines $\lambda$ and the target in~\eqref{eq:transport-target}.

\begin{proof}[Proof of Theorem~\ref{thm:transport}]
We prove the theorem in four steps. We first bound the diffuse paths produced by a single-site pulse. We then average over the width-$m$ initial pulse and bound the derivative of the normalized state. These estimates give the susceptibility scaling. Finally, two nearby values of $\lambda$ give the finite-sample lower bound, and we compare with the separately normalized diffuse component.

\paragraph{Single-site path expansion.}
The update~\eqref{eq:transport-update} describes a walk in which the stream changes independently with probability $\lambda/T$ before each step. First take a pulse at a single site. Let $b$ be its no-change endpoint, put $a=(1-\lambda/T)^T$, and write
\[
 q_\lambda=a(e_b+v_\lambda),\qquad
 (v_\lambda)_i=\sum_{k=1}^T N_{i,k}u^k,\qquad
 u=\frac{\lambda}{T-\lambda}.
\]
Here $N_{i,k}$ counts paths ending at label $i$ with exactly $k$ changes. Once the first $k-1$ change times are fixed, the endpoint determines the last change time, so $N_{i,k}\leq\binom{T}{k-1}$. Summing over $k$ gives the coefficientwise majorant $u(1+u)^T$. All coefficients are nonnegative, so differentiating this majorant also bounds $v'_{\lambda}$ coefficientwise. Since $a\geq1-\lambda\geq1/2$, $(1+u)^T=a^{-1}\leq2$, $Tu\leq1$, and $u'=T/(T-\lambda)^2\leq2/T$,
\begin{align*}
 0\leq(v_\lambda)_i
   &\leq u(1+u)^T\leq2/T,\\
 0\leq(v'_\lambda)_i
   &\leq u'(1+u)^T+Tuu'(1+u)^{T-1}\leq8/T,\\
 \sum_i(v_\lambda)_i&=a^{-1}-1,\qquad
 \sum_i(v'_\lambda)_i=\frac{a^{-1}}{1-\lambda/T}\leq3.
\end{align*}
The derivative bound is explicit:
\[
 u'\bigl[(1+u)^T+Tu(1+u)^{T-1}\bigr]
 \leq\frac2T(2+2)=\frac8T,
\]
and $\sum_i v'_{\lambda,i}\leq2/(1-1/8)=16/7<3$. The second line of the display is exactly the derivative of the polynomial majorant.

\paragraph{Width-$m$ pulse and normalized derivative.}
Convolving with the uniform width-$m$ initial distribution gives $q_\lambda=a(f_m+v_{\lambda,m})$. Averaging translates preserves the $\ell^1$ bound and cannot increase the $\ell^\infty$ bound, so
\[
 \norm{v'_{\lambda,m}}_2^2
 \leq\norm{v'_{\lambda,m}}_\infty
       \norm{v'_{\lambda,m}}_1
 \leq\frac8T\,3=\frac{24}{T}.
\]
Because $f_m$ is nonnegative and $\norm{f_m}_2^2=1/m$,
\[
 \begin{aligned}
 \norm{f_m+v_{\lambda,m}}_2^2&\geq1/m,
 &\norm{v'_{\lambda,m}}_2^2&\leq24/T,\\
 \norm{x'_\lambda}_2^2&\leq24m/T.
\end{aligned}
\]
The scalar factor $a$ cancels after normalization. Differentiating the normalized vector projects $v'_{\lambda,m}/\norm{f_m+v_{\lambda,m}}_2$ onto the tangent space, which gives the displayed bound on $\norm{x'_\lambda}_2^2$.

\paragraph{Susceptibility and restricted hard pair.}
For every $\lambda$, each stream has $\Theta(T)$ reachable labels. One change yields $T$ distinct backward labels, while zero or two changes give $T$ distinct forward labels. The union generated by a width-$m$ pulse remains in an interval of length at most $2T+m$ per stream. All have positive probability. Since $Z=1$,
\[
 \frac1{4m}\leq S=\norm{q_\lambda}_2^2
 \leq\norm{q_\lambda}_\infty\leq\frac1m+\frac2T\leq\frac3{2m}.
\]
Also $\mu_T$ stays in a fixed compact subinterval of $(0,1)$, and
\[
 \mu_T'(\lambda)=(1-2\lambda/T)^{T-1}\in[e^{-1},1].
\]
Consequently $\kappa^{-2}=WS=\Theta(T/m)$ and $\chi=S\sum_i(\phi_i-\mu_T)^2=\Theta(T/m)$. The inequality $|\mu_T'|^2\leq\chi\norm{x'_\lambda}_2^2$ gives the lower bound on the restricted susceptibility. The uniform lower bound on $|\mu_T'|$ together with $\norm{x'_\lambda}_2^2\leq24m/T$ gives the matching upper bound. This proves all three scalings in~\eqref{eq:transport-chi}.

\paragraph{Finite-sample lower and upper bounds.}
For the finite-sample lower bound choose $\lambda_\pm=3/8\pm3\varepsilon$. These lie in $[1/4,1/2]$ for the stated accuracy range, and their targets differ by at least $6\varepsilon/e>2\varepsilon$. The squared trace distance obeys
\[
 \begin{aligned}
 1-f&\leq\norm{x_{\lambda_+}-x_{\lambda_-}}_2^2
 \leq864\frac mT\varepsilon^2\leq\frac3{32},\\
 f&=\abs{\ip{x_{\lambda_+}}{x_{\lambda_-}}}^2.
 \end{aligned}
\]
Therefore $-\log f\leq1728(m/T)\varepsilon^2$. For two pure states with squared fidelity $f$, Helstrom's bound gives the same binary-testing inequality as~\eqref{eq:helstrom}, so
\[
 K^*\geq\frac{T}{1728m\varepsilon^2}
              \log\frac1{4\beta(1-\beta)}.
\]
For $0<\beta\leq1/3$ this logarithm is comparable to $\log(1/\beta)$. For the upper bound, apply Lemma~\ref{lem:public}: the uniform bounds $\kappa^{-2},\chi=O(T/m)$ make a block of $O(T/(m\varepsilon^2))$ copies succeed with probability at least $3/4$ on the whole family. A median of $O(\log(1/\beta))$ independent blocks gives the required confidence. The family need not lie in a single neighborhood~\eqref{eq:neighborhood}.

\paragraph{Alternative component interface.}
The conditionally normalized diffuse component is $r_\lambda=a v_{\lambda,m}/(1-a)$. Since $1-a\geq1-e^{-1/4}$, its largest entry is $O(T^{-1})$, and its support has $O(T)$ labels. It has unit mass, so $\norm{r_\lambda}_2^2\leq\norm{r_\lambda}_\infty\norm{r_\lambda}_1=O(T^{-1})$. Multiplying by the support size gives $\kappa_r^{-2}=O(1)$, and the same estimate gives $\chi_\phi(r_\lambda)=O(1)$ for the bounded target. Thus the same estimator has dimension-independent cost once the component mass is supplied. An additive $O(\varepsilon)$ error in that mass contributes only $O(\varepsilon)$ to the final target error. In this particular one-parameter family, however, the mass $a_\lambda$ alone already identifies $\lambda$ and hence the target. This is an interface comparison, not a claim that component separation is necessary when that scalar is available.
\end{proof}

\subsection{Numerical representations}\label{app:dictionaries}

These short pieces map common numerical outputs to the notation of the
paper.

\subsubsection{Modal records for linear systems}
To quantify this reduction, let $N\geq2$ and choose an orthonormal basis $h_1,\ldots,h_d$ of $x_0^\perp$, where $d=N-1$. The observables $\psi_j=\ell_0h_j$ have sensitivity vectors $g_{\psi_j}=h_j$. For the particular target $\phi$, order these directions by decreasing $a_j^2=\abs{\ip{g_\phi}{h_j}}^2$. This gives:
\begin{equation}\label{eq:ordered-record-tail}
 \chi_{\mathrm{rec},0}^{(k)}
   =\sum_{j>k}a_j^2
   \leq\left(1-\frac{k}{N-1}\right)\chi_0.
\end{equation}
For $\abs{\phi_i}\leq1$, positivity gives $S_0\leq Z_0^2$ and $\chi_0\leq4N$. Proposition~\ref{prop:extra-records} implies that the copy budget for the first $k$ records is:
\begin{equation}\label{eq:ordered-record-cost}
 K_{\mathrm{rec}}^{*,(k)}
   =O\!\left(\frac{N-1-k}{\varepsilon^2}\log\frac1\beta\right),
 \qquad k<N-1,
\end{equation}
This bound holds when $\chi_{\mathrm{rec},0}^{(k)}>0$ and
$\varepsilon\leq\kappa_0\sqrt{\chi_{\mathrm{rec},0}^{(k)}}/32$.
At $k=N-1$, the records determine the target exactly.

\subsubsection{Density evolution and mixing}
For a mass-conserving Fokker--Planck density $p_T$ on a finite-volume domain, set $q=p_T$, use volume measure, and choose a bounded regional or terminal observable $\phi$. Then $Z=1$ and $S=\norm{p_T}_2^2$. Uniform density bounds control the representation factor, but can depend on physical dimension, time, and diffusivity even when grid-independent.

For a positive concentration $q=c(x,T)$ with mass $M$, use $Z=M$ and $S=\norm{c}_2^2$. Periodic diffusion with divergence-free advection and conserved mass satisfies $d\norm{c}_2^2/dt=-2D\norm{\nabla c}_2^2\leq0$. These observables include mass fractions and scaled centers of concentration.

\subsubsection{Markov probabilities}

For a continuous-time Markov chain with probabilities $p_i(T)>0$ on $N$ states, choose $q_i=p_i(T)$, $\nu_i=1$, and a bounded reward $\phi_i$. Then $Z=1$, $S=\sum_i p_i^2$, and
\[
 \kappa^{-2}=NS,\qquad \chi=S\sum_i(\phi_i-\mu)^2.
\]
A distribution concentrated on $m\ll N$ states can give $\kappa^{-2}\asymp N/m$.

\subsubsection{Terminal payoffs}
For terminal-density pricing on a finite domain $D$, take $q=p_T$, $\nu$ equal to volume, and $\phi=g/B$ for $\abs{g}\leq B$. With $Z_D=\int_Dp_T$ and $S_D=\int_Dp_T^2$, $\chi=(S_D/Z_D^2)\int_D(\phi-\mu_D)^2\,\mathrm d\nu$. The truncated discounted payoff is $e^{-r_fT}BZ_D\mu_D$. The forward-density algorithm in~\cite{guseynov} uses payoff and uniform-reference overlaps to recover the price, with sampling-error factors that depend on those overlaps.

\subsubsection{Elliptic solutions}
Positive elliptic solutions describe normalized spatial averages when physically appropriate. For $-q''=2$ on $(0,1)$ with zero boundary values, $q=x(1-x)>0$ almost everywhere. With $\phi=x$ and Lebesgue measure,
\[
 Z=1/6,\quad S=1/30,\quad\kappa=\sqrt{5/6},\quad\chi=1/10.
\]

\subsubsection{Linear systems}
For an inverse-positive linear system $q=A^{-1}b>0$,
\[
 F_\phi(q)=\frac{(\nu\odot\phi)^TA^{-1}b}{\nu^TA^{-1}b}.
\]
Matrix condition number and readout susceptibility are not the same, as we discussed earlier. $A=I$ allows every positive $q$, while an ill-conditioned system can have a uniform solution.

\subsubsection{Finite-element coefficients}
For FEM, distinguish the physical $L^2$ norm from the Euclidean norm of encoded coefficients. If $u_h=\sum_i c_i\psi_i$, $c_i>0$, and the basis functions are nonnegative with $b_i=\int\psi_i>0$, an exact discrete translation for $\abs{r}\leq1$ is
\[
 \nu_i=b_i^2,\qquad q_i=c_i/b_i,\qquad
 \phi_i=\frac{\int r\psi_i}{b_i}.
\]
Then $\abs{\phi_i}\leq1$, $x_q=c/\norm{c}_2$, and $F_\phi(q)=\int r u_h/\int u_h$. The unnormalized functional in~\cite{montanaro} additionally uses recovered scale and coherent access.

\subsection{Signed and complex amplitude ratios}
\label{app:general-ratio}

This is an extension of the positive-real theorem, by relaxing the positivity assumption. The same ratio geometry applies to signed and complex amplitudes, including zeros, sign changes, and spatially varying phases, but the target and error are now complex-valued. The local copy bound requires a nonzero denominator at the reference state.

Let $H$ be a separable real or complex Hilbert space, with inner product conjugate-linear in its first slot. For known $a,b\in H$, $a\ne0$, consider, on the domain $\ip{a}{q}\ne0$,
\begin{equation}\label{eq:general-ratio-chi}
 \begin{aligned}
 F(q)&=\frac{\ip{b}{q}}{\ip{a}{q}},\\
 \chi_{a,b}(q)&=
 \frac{\norm{q}^{2}\norm{b-\overline{F(q)}a}^{2}}
      {\abs{\ip{a}{q}}^{2}}.
 \end{aligned}
\end{equation}
Both quantities are invariant under nonzero rescaling, including global sign or phase. For signed $q$, replacing the linear denominator by $\norm{q}_1$ would destroy invariance under global sign changes, making that target generally unidentifiable from the state alone. When $H=L^2(\nu)$ on a finite-measure space, setting $a=\one$ and real $b=\phi$ recovers the functional and susceptibility of the positive experiment.

Put $v=a/\norm{a}$, $w=b/\norm{a}$, and $x=q/\norm{q}$. The residual $r=w-\overline{F(x)}v$ obeys $\ip{r}{x}=0$, and
\[
 dF_x(u)=\frac{\ip{r}{u}}{\ip{v}{x}},\qquad
 \sup_{u\perp x,\,\norm{u}=1}\abs{dF_x(u)}^2=\chi_{a,b}(x).
\]
Therefore $\chi$ is the squared norm of the  derivative  over pure-state tangent directions orthogonal to $x$.

Choose a known unit center $x_0$ with $\kappa_0=\abs{\ip{v}{x_0}}>0$ and $\chi_0=\chi_{a,b}(x_0)>0$. Fix its phase so that $\ip{v}{x_0}=\kappa_0$, and define the neighbourghood
\begin{equation}\label{eq:general-ratio-ball}
 \mathcal N^{\rm rat}_0=
 \left\{[x]:\norm{x}=1,\quad
 \min_{\abs{\zeta}=1}\norm{\zeta x-x_0}\leq\kappa_0/4\right\}.
\end{equation}
Let $K^*_{a,b}$ be the smallest copy count under which one procedure can estimate $F(x)$ to error at most $\varepsilon$ with failure probability at most $\beta$.

\begin{theorem}[General amplitude-ratio readout]\label{thm:general-ratio}
If $0<\varepsilon\leq\kappa_0\sqrt{\chi_0}/12$ and $0<\beta\leq1/3$, then
\begin{equation}\label{eq:general-ratio-bounds}
 \begin{split}
 \frac{\chi_0}{10\varepsilon^2}
       \log\frac{1}{4\beta(1-\beta)}
 &\leq K^*_{a,b}(\varepsilon,\beta;\mathcal N^{\rm rat}_0)\\
 &\leq s m_{\rm rat}(\beta)\left\lceil\frac{256\chi_0}{\varepsilon^2}\right\rceil,
 \end{split}
\end{equation}
where $s=2$ for real $H$ and $s=3$ for complex $H$, and $m_{\rm rat}(\beta)$ is the smallest odd integer at least $18\log(2/\beta)$. Consequently $K^*_{a,b}(\varepsilon,\beta;\mathcal N^{\rm rat}_0) =\Theta(\chi_0\varepsilon^{-2}\log(1/\beta))$.
\end{theorem}

\begin{proof}
The proof has four parts. We first construct fixed rank-two settings and compute their moments. We then obtain a pointwise ratio bound, make it uniform on the projective neighborhood, and construct a close pair of alternatives for the converse.

\paragraph{Fixed settings.}
Use the reference direction $v$ to supply the denominator and subtract the reference ratio $c=\ip{w}{v}$ from the target. Define $z=w-\overline c v$, so $z\perp v$. The real and imaginary parts of the remaining overlap are measured by the two bounded Hermitian observables
\begin{equation}\label{eq:general-ratio-settings}
 \begin{split}
 D&=\ket v\bra v,\\
 C_R&=(\ket v\bra z+\ket z\bra v)/2,\\
 C_I&=(\ket v\bra z-\ket z\bra v)/(2i).
 \end{split}
\end{equation}
Omit $C_I$ in the real case. The operator $D$ has outcomes $0,1$, while $C_R,C_I$ have outcomes $0,\pm\norm{z}/2$. With $\kappa=\abs{\ip{v}{x}}$ and $\Delta=F(x)-c$, direct calculation gives
\begin{equation}\label{eq:general-ratio-moments}
 \begin{split}
 \E[D]&=\kappa^2,\qquad \E[C_R]+i\E[C_I]=\kappa^2\Delta,\\
 \chi&=(\norm{z}^2+\abs{\Delta}^2)/\kappa^2,\\
 \E[C_R^2]&=\E[C_I^2]=\kappa^4\chi/4.
 \end{split}
\end{equation}
Measure each setting on an independent $n$-copy batch and return $T_n=c+(\overline C_R+i\overline C_I)/\overline D$, with fallback $c$ when $\overline D=0$.

\paragraph{Pointwise bound.}

Let $G=\{\overline D\geq\kappa^2/2\}$ and $V=\overline C_R+i\overline C_I-\Delta\overline D$. Independence, $\E[V]=0$, and~\eqref{eq:general-ratio-moments} imply
\[
 \begin{aligned}
 \PP(G^c)&\leq\frac4{n\kappa^2},\\
 \E[\abs{V}^2]&\leq\frac{\kappa^4\chi}{2n}
     +\frac{\abs{\Delta}^2\kappa^2}{n}
 \leq\frac{3\kappa^4\chi}{2n}.
 \end{aligned}
\]
Given that $T_n-F=V/\overline D$ on $G$, Markov's inequality gives:
\begin{equation}\label{eq:general-ratio-failure}
 \PP[\abs{T_n-F}>\eta]
 \leq\frac4{n\kappa^2}+\frac{6\chi}{n\eta^2}.
\end{equation}

\paragraph{Uniformity on the ball.}
For uniformity choose a phase-aligned representative with $h=\norm{x-x_0}\leq\kappa_0/4$, and put $F_0=F(x_0)$, $r_0=w-\overline{F_0}v$. As $\ip{r_0}{x_0}=0$ and $\norm{r_0}=\kappa_0\sqrt{\chi_0}$,
\begin{equation}\label{eq:general-ratio-uniform}
 \begin{split}
 \kappa&\geq3\kappa_0/4,\\
 \abs{F(x)-F_0}
 &=\frac{\abs{\ip{r_0}{x-x_0}}}{\abs{\ip{v}{x}}}
 \leq\frac43\sqrt{\chi_0}\,h,\\
 \sqrt{\chi(x)}
 &\leq\frac{\norm{r_0}+\abs{F(x)-F_0}}{3\kappa_0/4}
 \leq\frac{16}{9}\sqrt{\chi_0}.
 \end{split}
\end{equation}
The accuracy restriction gives $\kappa_0^{-2}\leq\chi_0/(144\varepsilon^2)$. Thus~\eqref{eq:general-ratio-failure} with $\eta=\varepsilon/\sqrt2$ is at most $38\chi_0/(n\varepsilon^2)<1/3$ for $n=\lceil256\chi_0/\varepsilon^2\rceil$. Take coordinatewise medians of $m_{\rm rat}(\beta)$ independent complex blocks, or the ordinary median for real blocks. Each coordinate fails its $\varepsilon/\sqrt2$ tolerance with probability at most $e^{-m_{\rm rat}(\beta)/18}$ by Hoeffding's inequality~\cite{hoeffding}. A union bound gives complex absolute error at most $\varepsilon$ except with probability $2e^{-m_{\rm rat}(\beta)/18}\leq\beta$.

\paragraph{A close hard pair.}
For the lower bound, set $u=r_0/\norm{r_0}$, $t=3\varepsilon/(2\sqrt{\chi_0})$, and $x_\pm=(x_0\pm tu)/\sqrt{1+t^2}$. Orthogonality gives $\norm{x_\pm-x_0}\leq t\leq\kappa_0/8$, so both belong to~\eqref{eq:general-ratio-ball}. For $c_1=\ip{v}{u}$, $\abs{c_1}\leq1$, cancellation gives
\begin{equation}\label{eq:general-ratio-gap}
 \begin{split}
 F(x_+)-F(x_-)
 &=\frac{2t\norm{r_0}\kappa_0}{\kappa_0^2-t^2c_1^2},\\
 \abs{F(x_+)-F(x_-)}
 &\geq\frac{3\varepsilon}{1+1/64}>2\varepsilon.
 \end{split}
\end{equation}
And
\begin{equation}\label{eq:general-ratio-fidelity}
 \begin{split}
 f&=\abs{\ip{x_+}{x_-}}^2
   =\left(\frac{1-t^2}{1+t^2}\right)^2,\\
 -\log f&\leq\frac{4t^2}{1-t^2}
 \leq\frac{64\varepsilon^2}{7\chi_0}
 <\frac{10\varepsilon^2}{\chi_0}.
 \end{split}
\end{equation}
Because the two target values differ by more than $2\varepsilon$, an estimator accurate to $\varepsilon$ on both alternatives gives a binary test by assigning its output to the nearer target value. Padding unused inputs turns any protocol capped at $K$ copies into a POVM on $K$ copies. The two pure states then have squared fidelity $f^K$, so Helstrom's minimum average testing error is $(1-\sqrt{1-f^K})/2$. Therefore, for a probability of failure $\beta$ gives :
\[f^K\leq4\beta(1-\beta)\] which gives the stated lower bound. For $\beta\leq1/3$ its logarithm is comparable to $\log(1/\beta)$, and $\chi_0/\varepsilon^2\geq144$ absorbs the constants.
\end{proof}

\subsection{Why finite measure is required}\label{app:infinite}

The finite-measure assumption prevents diffuse mass from having negligible
$L^2$ norm but finite $L^1$ mass. On $[0,\infty)$ with Lebesgue measure,
let
\[
 q_0(x)=e^{-x},\qquad
 \phi(x)=e^{-1}\one_{[0,1]}(x)-\one_{[1,2]}(x).
\]
Then $Z_0=1$, $S_0=1/2$, $\mu_0=0$, and the formal susceptibility is
$\chi_0=(1+e^{-2})/2<\infty$. Fix a positive state-norm radius $r_0$.
Choose $\delta>0$ so that $q_\delta=q_0+\delta\one_{[0,1]}$ has normalized
state within $r_0/2$ of $x_0$. Its target is
$\mu_\delta=\delta/[e(1+\delta)]>0$. Now set
\[
 q_{\delta,L}=q_\delta+\frac{1+\delta}{L}\one_{[3,3+L]}.
\]
The added term has $L^2$ norm $(1+\delta)/\sqrt L$, so the normalized
states converge to $x_{q_\delta}$, while the target denominator doubles and
$F(q_{\delta,L})=\mu_\delta/2$. For fixed
$0<\varepsilon<\mu_\delta/4$, the target gap exceeds $2\varepsilon$ while
the fidelity tends to one. By Helstrom's bound, it excludes every finite
uniform copy cap on the radius-$r_0$ neighborhood for failure below $1/2$.

\end{document}